\documentclass[12pt]{article}
\usepackage{amsmath}
\usepackage{graphicx,psfrag,epsf}
\usepackage{enumerate}
\usepackage[sort,numbers, compress]{natbib}
\usepackage{xcolor}  
\usepackage{physics}  
\usepackage{amssymb} 
\usepackage{dsfont}
\usepackage{amssymb,amsmath,amsthm}
\newtheorem{theorem}{Theorem}
\usepackage{mathtools} 
\usepackage{graphicx}
\usepackage{booktabs}
\usepackage[linesnumbered,ruled,vlined]{algorithm2e}
\newtheorem{lemma}{Lemma}
\newtheorem{corollary}{Corollary}
\usepackage{dsfont}  
\usepackage{hyperref}

\newcommand{\blind}{0}

\begin{document}

\def\spacingset#1{\renewcommand{\baselinestretch}%
{#1}\small\normalsize} \spacingset{1}


\if0\blind
{
  \title{\bf Testing for a difference in means of a single feature after clustering}
  \author{Yiqun T. Chen \hspace{.2cm}\\
    Department of Biomedical Data Science, Stanford University\\
    and \\
    Lucy L. Gao \\
    Department of Statistics, University of British Columbia}
  \maketitle
} \fi

\if1\blind
{
  \bigskip
  \bigskip
  \bigskip
  \begin{center}
    {\LARGE\bf Title}
\end{center}
  \medskip
} \fi

\bigskip
\begin{abstract}
For many applications, it is critical to interpret and validate groups of observations obtained via clustering. A common validation approach involves testing differences in feature means between observations in two estimated clusters. In this setting, classical hypothesis tests lead to an inflated Type I error rate. To overcome this problem, we propose a new test for the difference in means in a single feature between a pair of clusters obtained using hierarchical or $k$-means clustering. The test based on the proposed $p$-value controls the selective Type I error rate in finite samples and can be efficiently computed. We further illustrate the validity and power of our proposal in simulation and demonstrate its use on single-cell RNA-sequencing data. 
\end{abstract}

\noindent%

\spacingset{1.45}

\section{Introduction}
\label{sec:introduction}

Clustering algorithms, a collection of computational tools designed to group unlabelled data, are ubiquitously applied across fields to preprocess, visualize, and compress large data sets~\citep{jaeger2022cluster}. It is often of interest to interpret and validate the results from clustering a data set: for instance, in the context of single-cell RNA sequencing (scRNA-seq), researchers often cluster cells based on their gene expression profiles, and want to interpret the resulting clusters as categorical measures of an unobserved aspect of the cells' biological state, such as cell type~\citep{grun2015single, aizarani2019human}. Similarly, in market segmentation, an analyst might cluster customers according to their measurable characteristics such as age, gender, and spending habits, and subsequently assign each resulting cluster a descriptive label (e.g. ``outdoorsy customers") to inform market design~\citep{leisch2018market}. 

Here, we consider how to determine which features are significantly different between two groups obtained via a clustering algorithm. Concretely, suppose that we applied a clustering algorithm to divide $n$ observations into $K$ non-overlapping groups based on $q$ features. For a pair of groups and for a feature $j \in \{1, 2, \ldots, q\}$, we want to answer the question: ``\emph{How can we assess whether the population means of the $j$th feature are the same between the two groups?}" 

Answering this question yields valuable insights. First, identifying the subset of features that appear to have different population means across cluster pairs facilitates cluster interpretation. For instance, in scRNA-seq, if the data suggests that the two cell clusters have different population mean expression levels for known marker genes of specific cell subtypes (e.g., helper T cells and killer T cells), this supports interpreting these cell clusters as the corresponding subtypes. Second, answers to this question could assist in evaluating the validity and generalizability of the obtained clusters. Given that clustering algorithms always output distinct clusters --- even when applied to observations from a single population --- observing at least one feature with population means across clusters increases our confidence in the resulting clusters, as well as the potential for generalizing our clustering results to new independent data sets. 

To ascertain whether the population mean of each feature is the same between groups, applying a classical test for equality of means between two populations (e.g., the two-sample $t$-test) for each feature and cluster pair might seem intuitive. However, such an approach ignores the fact that the null hypothesis of equal population means of a given feature between two clusters depends on the data used for testing, since the clusters are estimated on the same data. This leads to a failure to control the selective Type I error rate \citep{fithian2014optimal}; that is, the probability of falsely rejecting the null hypothesis, given that we chose to test it. Furthermore, sample splitting does not provide an adequate solution in this context, as clustering a subset of the observations does not directly lead to cluster assignments for the remaining observations; detailed discussion is available in~\citet{gao2022selective, chen2022selective, neufeldinference2022}.

In this paper, we develop a finite-sample selective inference framework \citep{fithian2014optimal} for testing for a difference in means of a single feature in two clusters, under a multivariate Gaussian assumption. In short, to account for the fact that the clusters are estimated using the same data used for testing, we condition on the event that the clustering algorithm outputs a particular partition of the observations, thereby controlling the selective Type I error rate. In the special case of $k$-means clustering and hierarchical clustering --- two of the most popular forms of clustering --- we provide an analytical characterization of the conditioning set that enables efficient and exact computation of our proposed $p$-value. 

Our work is closely related to \citet{gao2022selective} and \citet{chen2022selective} and amounts to extending their selective inference framework for testing the difference in \emph{vector means} to \emph{individual feature means}. While this manuscript is under preparation, \citet{hivert2022post} proposed a related selective inference framework for the difference between the mean of a single feature in two clusters. Compared to their work, our proposal (i) does not assume that the features used for clustering are independent; and (ii) computes the $p$-value exactly with a computationally efficient algorithm, rather than approximating the $p$-value via Monte Carlo sampling. Methods developed in this paper are implemented in the \texttt{R}
package \texttt{CADET} (Clustering And Differential Expression Testing) available at \url{https://github.com/yiqunchen/CADET}. Data and code for reproducing the results in this paper can be found at \url{https://github.com/yiqunchen/CADET-experiments}.

The rest of our paper is organized as follows. In Section~\ref{sec:motivation}, we review the problem of testing for a difference in means after clustering. In Sections~\ref{sec:select} and \ref{sec:truncation}, we propose tests that control the selective Type I error rate when testing for a difference in means after hierarchical or $k$-means clustering, and provide a computationally-efficient approach to compute the $p$-values corresponding to our proposal. We evaluate our proposal in a simulation study in Section~\ref{sec:simulation} and apply our proposal to real datasets in Section~\ref{sec:application}. Proofs and additional results are relegated to the Appendix. 

Throughout this paper, we will use the following notational conventions. $I_n$, $0_n$ and $ \mathds{1}\{\cdot\}$ denote the $n$-dimensional identity matrix, $n$-vector of zeros, and the indicator function, respectively. For a matrix ${A}$, ${A}_i$ denotes the $i$th row and $A_{ij}$ denotes the $(i,j)$th entry. For a vector $\nu \in \mathbb{R}^n$, $\Vert\nu\Vert_2$ denotes its $\ell_2$ norm, and ${\Pi}_\nu^\perp$ is the projection matrix onto the orthogonal complement of $\nu$, i.e., ${\Pi}_\nu^\perp = I_n-\nu\nu^{\top}/\Vert\nu\Vert_2^2$.

\section{Motivation}
\label{sec:motivation}
\subsection{Model and data} 
Let $x \in \mathbb{R}^{n \times q}$ be a data matrix with $n$ observations of $q$ features. For $\mu \in \mathbb{R}^{n \times q}$ with unknown rows $\mu_i \in \mathbb{R}^q$ and known, positive-definite $\Sigma \in \mathbb{R}^{q \times q}$, we assume that $x$ is a realization of a random matrix $X$, where rows of $X$ are independent and drawn from a multivariate normal distribution:
\begin{equation} 
X_i \sim N_q(\mu_i, \Sigma), \quad i = 1, 2, \ldots, n. \label{eq:mod}
\end{equation} 

\subsection{Testing two pre-defined groups} 
\label{sec:pre-defined}

Let $j \in \{1, 2, \ldots, q\}$. For any $G \subset \{1, 2, \ldots, n\}$, let
\begin{equation} 
\bar{\mu}_{Gj} = \sum \limits_{i \in G} \sum \limits_{j=1}^q \mu_{ij}/|G|, 
\end{equation} 
be the mean of the $j$th feature in the group $G$. Consider using $x$ to test the null hypothesis that there is no difference in the mean of the $j$th feature across two \emph{pre-defined}, non-overlapping groups $G, G' \subset \{1, 2, \ldots, n\}$, i.e. 
\begin{equation} 
H_{0j}: \bar{\mu}_{Gj} = \bar{\mu}_{G'j} \text{ versus } H_{1j}: \bar{\mu}_{Gj} \neq \bar{\mu}_{G'j}.
\label{eq:hyp-z}
\end{equation} 
This is equivalent to testing $H_{0j}: [\mu^T \nu]_j = 0$ versus $H_{1j}: [\mu^T \nu]_j \neq 0$, where $\nu$ is the $n$-vector with $i$th element given by $\mathds{1} \{i \in G\}/|G| - \mathds{1} \{i \in G'\}/|G'|$. Since $G$ and $G'$ are chosen independently of the data used for testing, we could test $H_{0j}$ by applying the two-sample $Z$-test, with $p$-value given by 
$\mathbb{P}_{H_{0j}} \left ( | [X^T \nu]_j | \geq | [x^T \nu]_j| \right )$. Under \eqref{eq:mod}, this amounts to computing 
$1 - 2 \Phi \big (\big |[\nu^T x]_j \big | / \left (\|\nu\|_2^2 \Sigma_{jj} \right ) \big )$, where $\Phi(\cdot)$ is the cumulative distribution function of the standard normal distribution.

\subsection{What changes when the groups are estimated clusters?}

\label{sec:estimated-groups}

Let  $\mathcal{C}(\cdot)$ be a clustering algorithm that takes in a data matrix $x$ with $n$ rows and outputs a partition of $\{1,2, \ldots, n\}$. Suppose that we now want to use $x$ to test the null hypothesis that there is no difference in the mean of the $j$th feature across two groups obtained by applying $\mathcal{C}(\cdot)$ to $x$, i.e. 
\begin{equation} 
\hat H_{0j}: \bar{\mu}_{\hat G j } =  \bar{\mu}_{\hat G' j } \text{ versus } \hat H_{1j}: \bar{\mu}_{\hat G j } \neq  \bar{\mu}_{\hat G' j },
\label{eq:hyp}
\end{equation} 
where $\hat G, \hat G' \in \mathcal{C}(x)$ are a pair of estimated clusters. This is equivalent to testing $\hat H_{0j}: [\mu^T \hat \nu]_j = 0$ versus $\hat H_{1j}: [\mu^T \hat \nu]_j \neq 0$, where $\hat \nu$ is the $n$-vector with $i$th element given by 
\begin{equation} 
[\hat \nu]_i = \mathds{1} \{i \in \hat G\}/|\hat G| - \mathds{1} \{i \in \hat G'\}/|\hat G'|.
\end{equation} 
The challenge is that $\hat H_{0j}$ is a function of the data used to test it, because $\hat G$ and $\hat G'$ are estimated clusters in $\mathcal{C}(x)$. We could naively treat $\hat G$ and $\hat G'$ as pre-specified groups, and test $\hat H_{0j}$ by applying the two-sample $Z$-test as described in Section \ref{sec:pre-defined}. This would lead to the following $p$-value:
\begin{equation} 
p_{j, \text{naive}} = \mathbb{P}_{\hat H_{0j}} \left ( | [X^T \hat \nu]_j | \geq | [x^T \hat \nu]_j| \right )
\label{eq:naive}
\end{equation} 
That is, we could compare the values of $x_{ij}$ for $i \in \hat G \cup \hat G'$ to the distribution of $X_{ij}$ for $i \in \hat G \cup \hat G'$. However, because $\mathcal{C}(X)$ is random and dependent on $X$, the distribution of $X_{ij}$ stratified by $\mathcal{C}(X)$ can be far from the distribution of $X_{ij}$ stratified by $\mathcal{C}(x)$. Consequently, over repeated realizations of $X$, applying the two-sample $Z$-test to compare the means of two estimated clusters will lead to anti-conservative inference.

To illustrate this in an example, we simulate data from \eqref{eq:mod} with $n=150,q=10$, and 
{
\begin{align}
\label{eq:intro_toy_model}
{\mu}_i = (1 ~ 0_9)^T \text{ for } i \leq 50; ~ {\mu}_i = ( 0_9 ~ 1)^T \text{ for } i > 50, 
\end{align}}
with $\Sigma_{ij} = 1\{i=j\}+0.4\cdot 1\{i\neq j\}$
so that we have two equally sized true clusters that differ in the first and last features only. In each simulated data set, we apply $k$-means clustering to obtain two clusters, and test for a difference in the means of the estimated clusters for each of the eight features. Across all of the simulated data sets, there is no difference in the means of the estimated clusters for features 2--8 under \eqref{eq:intro_toy_model}; see Figure~\ref{fig:toy_example}(a). Nevertheless, there can be a substantial difference in the empirical means of the estimated clusters for features 2--8; see Figure~\ref{fig:toy_example}(b).  Thus, over 1,500 simulated data sets, and over features 2--8, the $p$-values from the two-sample $Z$-test appear far smaller than a Uniform$(0, 1)$ distribution; by contrast, the $p$-values from our proposed test (Section \ref{sec:select}) follow a Uniform$(0, 1)$ distribution (Figure \ref{fig:toy_example}(c)).


\begin{figure*}[htbp!]
  \centering
  \includegraphics[width=0.9\linewidth]{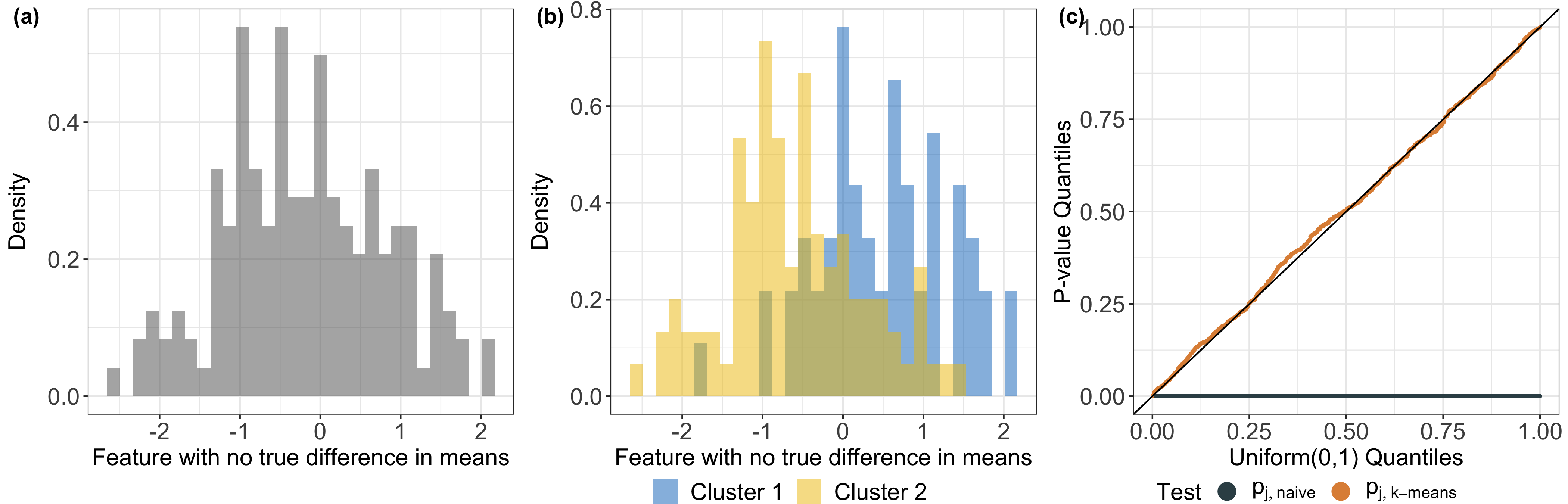}
\caption{We simulated one data set from \eqref{eq:mod} with ${\mu}$ and $\Sigma$ specified in \eqref{eq:intro_toy_model}. \textit{(a) } Empirical distribution of feature 2 based on the simulated data set. \textit{(b): } We apply $k$-means clustering to obtain two clusters and plot the empirical distribution of feature 2 stratified by the clusters. \textit{(c): } Quantile-quantile plot of the $p$-values from the two-sample $Z$-test applied to the estimated clusters (defined in \eqref{eq:naive}) and our proposal (defined in \eqref{eq:selective}), aggregated over 1,500 simulated data sets and over features 2--8 (i.e., the features with no true difference in means across \emph{any} pairs of estimated clusters).}
\label{fig:toy_example}
\end{figure*}

\section{Selective inference for the mean of a single feature} 
\label{sec:select}

We will overcome the challenges discussed in Section \ref{sec:estimated-groups} by developing a selective inference framework \citep{fithian2014optimal} for testing the equality of the means of a single feature between two \emph{estimated} clusters. 

\subsection{The ``ideal'' p-value}

In the setting outlined in Section \ref{sec:estimated-groups}, we chose to test $\hat H_{0j}$ in \eqref{eq:hyp} using $x$ because $\hat G, \hat G' \in \mathcal{C}(x)$. Thus, \citet{fithian2014optimal} argues that we should apply a test that controls the \emph{selective Type I error rate} at level $\alpha$, which guarantees that the proportion of times we falsely reject a selected null hypothesis is controlled at level $\alpha$ over repeated realizations of $X$:
{
\begin{equation} 
\mathbb{P}_{\hat H_{0j}} \left ( \text{Reject } \hat H_{0j} \text{ at level } \alpha \mid \text{Choose to test } \hat H_{0j}  \right ) \leq \alpha. \label{eq:selective_type_1}
\end{equation} 
}

This motivates the following conditional version of the two-sample $Z$-test in \eqref{eq:naive} to test $\hat H_{0j}$: 
\begin{equation} 
\mathbb{P}_{\hat H_{0j}} \left ( |[X^T \hat \nu]_j| \geq |[x^T \hat \nu]_j| ~ \Big | ~ \mathcal{C}(X)=\mathcal{C}(x) \right ), \label{eq:ideal} 
\end{equation} 
where we conditioned on $\{\mathcal{C}(X)=\mathcal{C}(x)\}$ because the hypothesis of interest was chosen based on the clustering output. By the probability integral transform, rejecting $\hat H_{0j}$ if the $p$-value in \eqref{eq:ideal} is less than $\alpha$ controls the selective Type I error rate at level $\alpha$.

In practice, computing \eqref{eq:ideal} is challenging, as (i) the conditional distribution of $[X^T \hat \nu]_j$ depends on unknown parameters that are left unspecified by $\hat H_{0j}$; and (ii) the conditioning set $\{X \in \mathbb{R}^{n \times q}:\mathcal{C}(X)=\mathcal{C}(x)\}$ depends on the clustering algorithm $\mathcal{C}$ and could be highly non-trivial to characterize. In Section~\ref{section:truncated_pval}, we will overcome these two challenges by modifying \eqref{eq:ideal} to condition on extra information; this leads to a computationally tractable test that controls the selective Type I error rate when the clusters are obtained via hierarchical or $k$-means clustering. 
\subsection{Truncated Gaussian p-value} 
\label{section:truncated_pval}
To overcome the challenges in computing \eqref{eq:ideal}, we condition on additional events and compute: 
\begin{align} 
 \label{eq:selective} 
 p_{j, \text{selective}}  = \mathbb{P}_{\hat H_{0j}} \Big (  \big |[X^T \hat \nu]_j \big | \geq \big |[x^T \hat \nu]_j  \big | ~ \Big |~ \mathcal{C}(X) = \mathcal{C}(x), U(X) = U(x) \Big ),
\end{align}
where 
\begin{align} 
U(x) = x - \frac{\hat \nu \Sigma_{j}^T [x^T \hat \nu]_j }{\|\hat \nu\|_2^2 \Sigma_{jj}}. \label{eq:defU}
\end{align}
Compared to \eqref{eq:ideal}, we have conditioned on  $\left \{ U(X) = U(x) \right \}$. This choice does not sacrifice control of the selective Type I error rate (see Proposition 3 in \citet{fithian2014optimal}). Furthermore, we can rewrite $X$ in \eqref{eq:mod} as:  
{\begin{align} 
 X
&= \left ( X - \frac{\hat \nu \Sigma_{j}^T [X^T \hat \nu]_j }{\|\hat \nu\|_2^2 \Sigma_{jj}} \right ) + \left ( \frac{\hat \nu}{\|\hat \nu\|_2^2} \right ) \left (\frac{\Sigma_j}{\Sigma_{jj}} \right )^T [X^T \hat \nu]_j,\label{eq:decomp}
\end{align} }
where the first term is in the conditioning set of \eqref{eq:selective} and the second term depends on $X$ only through our test statistic $[X^T \hat \nu]_j$. It turns out that the two terms on the right-hand-side of \eqref{eq:decomp} are independent under model \eqref{eq:mod}. 
Thus, when evaluating the conditional probability in \eqref{eq:selective},  we only need to consider the randomness in $X$ coming from the \emph{scalar-valued} test statistic $[X^T \hat \nu]_j$, despite the fact that all of $X$ is involved in the conditioning event $\{\mathcal{C}(X)=\mathcal{C}(x)\}$. 

This intuition is formalized in the following result, which says that (i) computing \eqref{eq:selective} involves a truncated univariate normal distribution; and (ii) testing $\hat H_{0j}$ using \eqref{eq:selective} controls the selective Type I error rate. 
\begin{theorem} 
\label{thm:truncnorm}
Suppose that $j \in \{1, 2, \ldots, q\}$, $x$ is a realization from \eqref{eq:mod}, and $\mathbb{F}(t; \mu, \sigma, \mathcal{S})$ denotes the cumulative distribution function (CDF) of a $N(\mu, \sigma^2)$ random variable truncated to the set $\mathcal{S}$. Then, for $p_{j, selective}$ defined in \eqref{eq:selective}, we have that 
{\begin{align} 
\begin{split}
\label{eq:selective_cdf}
p_{j, \text{selective}} = 1 - \mathbb{F} \left (\big | [\hat \nu^T x]_j \big | ; 0, \Sigma_{jj}\|\hat \nu\|_2^2; \mathcal{\hat S}_j \right ) + \mathbb{F}\left (-\big | [\hat \nu^T x]_j \big | ; 0, \Sigma_{jj} \|\hat \nu\|_2^2;  \mathcal{\hat S}_j \right ),
\end{split}
\end{align}} 
where $x'(\phi,j) = x + (\phi - ( \bar{x}_{\hat{G}j} - \bar{x}_{\hat{G'}j})) \left ( \frac{ \hat\nu }{ \|\hat\nu\|_2^2 } \right ) \left ( \frac{\Sigma_j}{\Sigma_{jj}} \right )^T,$  and
{\begin{equation} 
\hat S_j = 
\left \{ \phi \in \mathbb{R}: C(x) = \mathcal{C}\left(x'(\phi,j)\right ) \right \}. \label{eq:defS}
\end{equation} }
Furthermore, the test that rejects $\hat H_{0j}: \bar{\mu}_{\hat Gj} = \bar{\mu}_{\hat G' j}$ when $p_{j, selective} \leq \alpha$ controls the selective type I error rate at level $\alpha$, in the sense of \eqref{eq:selective_type_1}.  
\end{theorem}

It follows from Theorem \ref{thm:truncnorm} that computing the selective $p$-value in \eqref{eq:selective} amounts to computing the truncation set in \eqref{eq:defS}. The next section is dedicated to understanding and computing this truncation set.

\section{The truncation set}
\label{sec:truncation}



\subsection{Intuition} 
\label{section:Intuition}

The set $\hat{S}_j$ defined in \eqref{eq:defS} represents the values of $\phi$ for which the clustering algorithm $\mathcal{C}$, when applied to $x'(\phi,j)$, yields the clustering output $C(x)$.  Here, ${x}'(\phi,j) = x + (\phi - ( \bar{x}_{\hat{G}j} - \bar{x}_{\hat{G}'j})) \left ( \frac{ \hat\nu }{ \|\hat\nu\|_2^2 } \right ) \left ( \frac{\Sigma_j}{\Sigma_{jj}} \right )^T$ can be interpreted as a perturbation to the observed data ${x}$. 

Figure~\ref{fig:x_phi} illustrates a realization of \eqref{eq:mod} with $n=30, q=2$, and a covariance matrix $\Sigma$ encoding moderate correlation (0.4) between any two features.  Panel (a) displays the observed data ${x}$, which corresponds to ${x}'(\phi)$ with $\phi={x}_1^T\hat\nu =-3$. Here, $\hat\nu$ was chosen to test the difference between $\hat{G}$ (shown in blue) and $\hat{G'}$ (shown in rosy brown), estimated from $k$-means clustering with $K=3$. Panels (b) and (c) of Figure~\ref{fig:x_phi} display ${x}'(\phi)$ with $\phi = 0$ and $\phi = 6$, respectively. In panel (b), with $\phi=0$, the blue and rosy brown clusters are ``pushed together'' in the first feature, resulting in ${x}'(\phi)_1^T \hat\nu =  0$; that is, there is no difference in empirical means between feature 1 (x-axis of panel (b)) of the two clusters under consideration. By contrast, in panel (c), with $\phi=-5$, the blue and rosy brown clusters are ``pulled apart'', which results in an increased distance between the first feature of the blue and rosy brown clusters. 

When put together, panels (a)--(c) reveal that (i) $\phi$ can be interpreted as the observed ``test statistic'' $x_j^T\hat\nu$ on $x'(\phi)$; (ii) varying $\phi$ only changes the values of clusters $\hat{G}$, $\hat{G}'$ and leaves the other clusters (e.g., the orange cluster in Figure~\ref{fig:x_phi}) alone; and (iii)  $\phi$ moves the observed difference in all features correlated with the feature being tested (e.g., feature 2 in Figure~\ref{fig:x_phi}); we visualize the magnitude of the changes in Figure~\ref{fig:x_phi}(d). In this case, the slope of the blue line is the correlation between the two features. 

\begin{figure*}[htbp!]
  \centering
  \includegraphics[width=\linewidth]{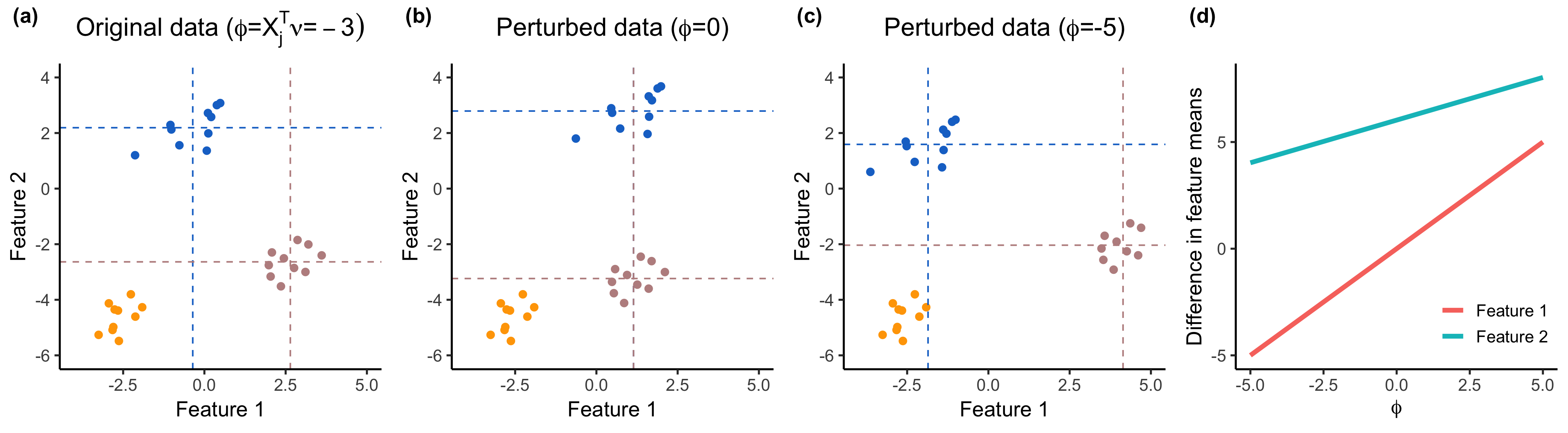}
\caption{One simulated data set generated from model \eqref{eq:mod} with ${\mu}_{i} = 1\qty{1\leq i\leq 10}[0,2.5]^T + 1\qty{11 \leq i \leq 20} [2.5, 0]^T + 1\qty{21\leq i\leq 30} [-2.5,-2.5]^T$ and $\Sigma = 0.2\cdot[1, 0.4; 0.4 ,1]$. \textit{(a):} The original data ${x}$ corresponds to $\phi = -3$. Applying $k$-means clustering with $K=3$ yields three clusters (rosy brown, blue, and orange). Here, $\hat{\nu}$ is chosen to test for a difference in means between $\hat{G}$ (blue) and $\hat{G}'$ (rosy brown). Empirical means for features 1 and 2 are displayed in dashed lines for $G$ and $G'$.  \textit{(b):} Perturbed data ${x}'(\phi, 1)$ at $\phi = 0$ results in no empirical mean difference for the first feature between the blue and rosy brown clusters. \textit{(c): } With ${x}'(\phi, 1)$ at $\phi = -5$, the mean difference for the first feature becomes more pronounced. \textit{(d):} The empirical difference in features 1 (red line) and 2 (blue line) as a function of $\phi$. The slope of the red line is 1 by the definition of $x'(\phi, j)$ and the slope of the blue line is $\Sigma_{12}/\Sigma_{11}$, where $\Sigma$ is the covariance matrix of the features.}
\label{fig:x_phi}
\end{figure*}


\subsection{Computing $\mathcal{\hat S}$ for hierarchical clustering}

We first review an important result from \citet{gao2022selective}. For any $w \in \mathbb{R}^q$, define the set
{
\begin{equation} 
\mathcal{S}(w) \equiv \left \{ \phi \in \mathbb{R}:\mathcal{C}\left( \left (I_n - \frac{\hat \nu \hat \nu^T}{\|\hat \nu\|_2^2} \right ) x  + \frac{ \phi \hat \nu w^T}{\|\hat \nu\|_2^2}   \right) = \mathcal{C}(x) \right \}.
\end{equation} 
}
\begin{theorem}[Gao et al. (2023)] 
\label{thm:hier}
Let $K > 1$ and consider applying hierarchical clustering to the squared Euclidean distance matrix and cutting the resulting dendrogram to get $K$ clusters. Then, for any realization $x$ from \eqref{eq:mod}, and any $w \in \mathbb{R}^q$, the set $\mathcal{S}(w)$ can be computed in at most $\mathcal{O}(n^2 + n \log(n))$ operations for single and average linkage, $\mathcal{O}(n^3 + n \log(n))$ operations for centroid linkage, and $\mathcal{O}(n^2 + n \log(n))$ operations for Ward linkage. 
\end{theorem} 
Theorem \ref{thm:hier} is a direct generalization of results in Section 3 of \citet{gao2022selective}. In short, Section 3.2 of \citet{gao2022selective} shows that $\mathcal{S}(w)$ is the intersection of $\mathcal{O}(n^2)$ sets, where $n$ is the number of observations.  Sections 3.3 and 3.4 of \citet{gao2022selective} further reveal that for average, Ward, centroid, and single linkage, each of the $\mathcal{O}(n^2)$ intersected sets are the solution sets to a quadratic inequality in $\phi$. Observing that we can take the intersection of the solution sets of $N$ quadratic inequalities in $\mathcal{O}(N \log N)$ operations, and carefully analyzing the number of operations needed to compute the coefficients of the quadratic inequalities using the squared Euclidean distance matrix, leads to the worst-case time complexities listed in Theorem \ref{thm:hier}. 

Since $\mathcal{\hat S}$  in \eqref{eq:defS} can be written as
$\mathcal{\hat S} = \mathcal{S}(\hat \nu, \Sigma_j/\Sigma_{jj})$;
it follows from Theorem \ref{thm:hier} that computing $\mathcal{\hat S}$ for hierarchical clustering requires the worst-case time complexities listed in Theorem \ref{thm:hier}. 
\subsection{Extensions to $k$-means clustering}

In this section, we extend the proposed $p$-value \eqref{eq:selective} to the $k$-means clustering algorithm and outline an efficient computational recipe. Because $k$-means clustering iteratively updates the cluster assignment~\citep{Lloyd2006-io}, characterizing $\{X:\mathcal{C}(X) =  \mathcal{C}(x)\}$, where $\mathcal{C}$ denotes the final clusters at convergence, may require enumerating possibly an exponential number of intermediate cluster assignments. Hence, following~\citet{chen2022selective}, we condition on \emph{all of the intermediate clusters} in the $k$-means clustering algorithm to arrive at this extension of the $p$-value in \eqref{eq:selective}:
{\begin{align} 
p_{j,\text{k-means}} 
&= \mathbb{P}_{{\hat{H}_{0j}}} \Bigg (  \big |[X^T \hat\nu]_j \big | \geq \big |[x^T \hat\nu]_j  \big | ~ \Bigg | ~ U(X) = U(x), \bigcap_{t=0}^{T}\bigcap_{i=1}^n\big\{ c_i^{(t)}(X) = c_i^{(t)}(x)\big\}  \Bigg  ),\label{eq:kmenas-selective} 
\end{align}}
where $c_i^{(t)}(\cdot)$ is the assigned cluster of observation $i$ at the $t$th iteration of the $k$-means algorithm, $U(\cdot)$ is defined in \eqref{eq:defU}, and $T$ is the total number of iterations run during $k$-means clustering. 

As in Section \ref{sec:select}, conditioning on additional information in \eqref{eq:kmenas-selective} still guarantees selective type I error control, and $p_{j,\text{k-means}}$ can be computed using a univariate truncated Gaussian distribution; this is formalized in Proposition~A1 in the Appendix.  Regarding computation, we can extend the ideas in Section 3 of \citet{chen2022selective} to efficiently compute the conditioning set in \eqref{eq:kmenas-selective}. The key idea is that we can recast the computation to solving $O(nTK)$ number of quadratic inequalities in $\phi$ and intersecting the resulting solution sets, taking $\mathcal{O}((n+q)KT+nKT\log(nKT))$ operations in total (see details in the Appendix).


\section{Simulation study}
\label{sec:simulation}

\subsection{Overview}
Throughout this section, we consider testing the null hypothesis $\hat H_{0j}: \bar{\mu}_{\hat G j } =  \bar{\mu}_{\hat G' j } \text{ versus } \hat H_{1j}: \bar{\mu}_{\hat G j } \neq  \bar{\mu}_{\hat G' j }$, where, unless otherwise stated, $\hat G$ and $\hat G'$ are a randomly-chosen pair of clusters from $k$-means or hierarchical clustering, and $j$ is a randomly-chosen feature. We consider the following $p$-values: $p_{j, \text{naive}}$ in \eqref{eq:naive}, $p_{j, \text{k-means}}$ in \eqref{eq:kmenas-selective}, $p_{j, \text{average}}$, $p_{j, \text{centroid}}$, and  $p_{j, \text{single}}$ (defined in \eqref{eq:selective} where $\mathcal{C}$ is hierarchical clustering with average, centroid, and single linkage, respectively). 

\subsection{Selective Type I error rate}

We generate data from \eqref{eq:mod} with $
{\mu}_i = (1 , 0_{q-1})^T \text{ for } i \leq 50; ~ {\mu}_i = ( 0_{q-1} ~ 1)^T \text{ for } i > 50 ;  \text{ and } \Sigma_{ij} = 1\{i=j\}+\rho\cdot 1\{i\neq j\}$; therefore, $\hat{H}_{0j}$ holds for any pair of estimated cluster and any feature $j = 2,\ldots, q-1$. We simulated 1,500 data sets with $q=10$ and $\rho = 0,0.4,0.8$.

For each simulated data set, we apply $k$-means clustering and hierarchical clustering with average, centroid, and single linkage to estimate three clusters. We then compute $p_{j, \text{naive}}$ (based on the output from $k$-means clustering), $p_{j, \text{k-means}}$, $p_{j, \text{average}}$, $p_{j, \text{centroid}}$, and $p_{j, \text{single}}$ for a randomly-chosen pair of clusters and a random feature between 2 and $q-1$. 

Figure~\ref{fig:sim_type_1} displays the observed $p$-value quantiles versus the Uniform(0,1) quantiles. We see that for all values of $q$ and $\rho$, (i) the naive $p$-values in \eqref{eq:naive} are stochastically smaller than a Uniform(0,1) random variable, suggesting that the test based on $p_{j,\text{naive}}$ leads to an inflated Type I error rate (the number of false rejections increases as the underlying feature correlation increases); (ii) tests based on $p_{j, \text{k-means}}$, $p_{j, \text{average}}$, $p_{j, \text{centroid}}$, and $p_{j, \text{single}}$ control the selective Type I error rate in the sense of \eqref{eq:selective_type_1}.

\begin{figure}[htbp!]
  \centering
\includegraphics[width=0.85\linewidth]{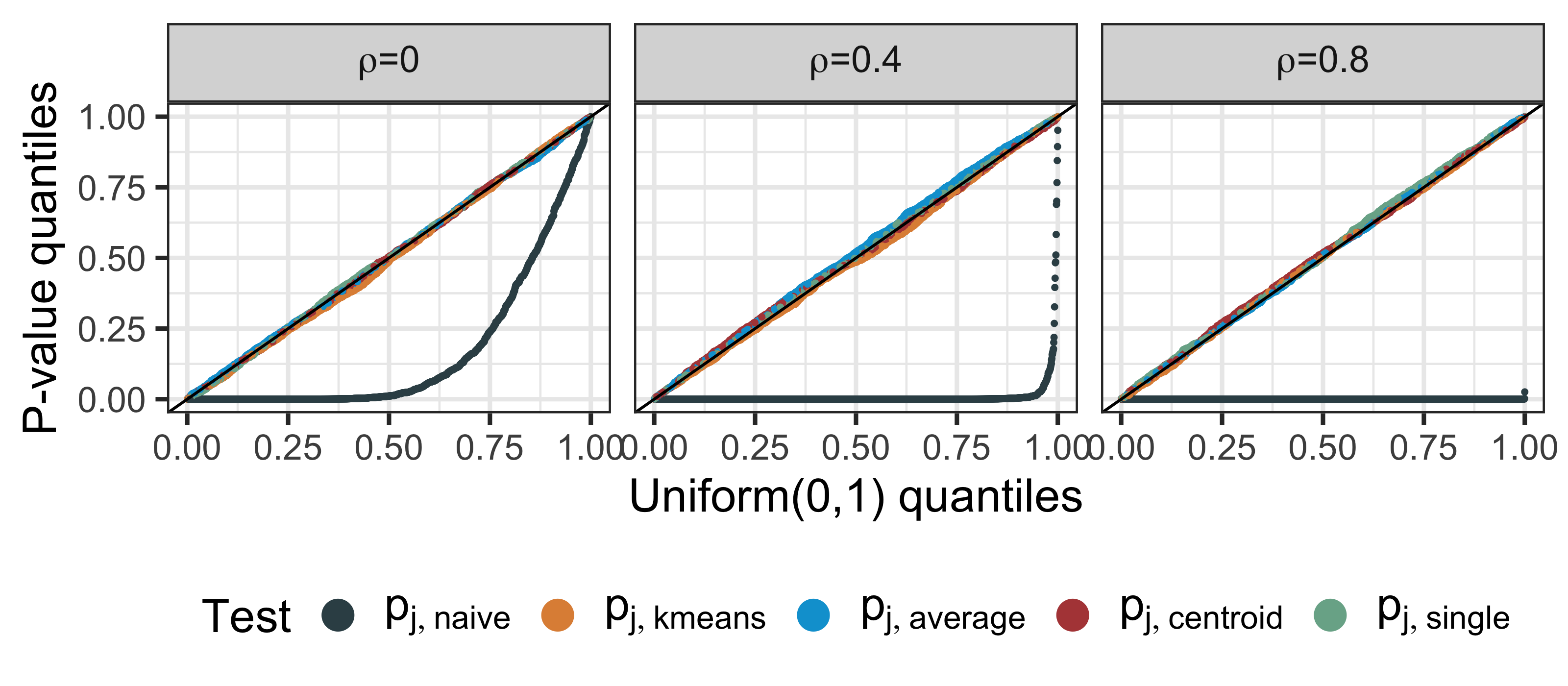}
\caption{Quantile-quantile plots for $p_{j, \text{naive}}$, $p_{j, \text{k-means}}$, $p_{j, \text{average}}$, $p_{j, \text{centroid}}$, and $p_{j, \text{single}}$ under \eqref{eq:mod} when $\hat{H}_{0j}$ holds.}
\label{fig:sim_type_1}
\end{figure}

\subsection{Conditional power and detection probability}
\label{subsec:power}
In this section, we demonstrate that tests based on our proposed $p$-values in Section~\ref{sec:select} have substantial power to reject $\hat{H}_{0j}$ when it does not hold. We generate data from \eqref{eq:mod} with three equally sized ``true clusters", ${G}_1 = \{1,\ldots,50\},{G}_2 = \{51,\ldots,100\},$ and ${G}_3 = \{101,\ldots,150\}$:
{\begin{align}
\begin{split}
\label{eq:power_model}
&{\mu}_i = (0_{\lfloor q/2 \rfloor} , -\delta_{\lceil q/2 \rceil})^T \text{ for } i \leq 50; \\ 
& {\mu}_i = 0_{q} \text{ for } 50 < i < 100; \\
&{\mu}_i  = (0_{\lfloor q/2 \rfloor} , \delta_{\lceil q/2 \rceil})^T  \text{ for } i \geq 100; \text{ and } \\ 
&\Sigma_{ij} = 1\{i=j\}+\rho\cdot 1\{i\neq j\}.
\end{split}
\end{align}}
 We simulated 2,000 data sets for each combination of $q=10$, $\rho = 0,0.4,0.8$, and $\delta=3,\ldots,8$. 

For each simulated data set, we computed $p_{j, \text{k-means}}$, $p_{j, \text{average}}$, $p_{j, \text{centroid}}$, and $p_{j, \text{single}}$ for a randomly-chosen pair of clusters and rejected $\hat{H}_{0j}: \bar{\mu}_{\hat G j } =  \bar{\mu}_{\hat G' j }$ if these $p$-values were less than $\alpha = 0.05$. Note that different clustering methods may estimate different clusters in a single data set, leading to different null hypotheses. Thus, our analysis evaluates both the \emph{conditional power} of the tests and the \emph{detection probability} of the employed clustering methods~\citep{gao2022selective,chen2022selective,Jewell2022-fn,Hyun2021-xm}. We define the conditional power as the probability of rejecting $\hat{H}_{0j}$ in \eqref{eq:hyp} \emph{given that} $\hat{G}$ and $\hat{G}'$ are true clusters. Given $M$ simulated data sets with true clusters $\qty{\mathcal{G}_1,\ldots,  \mathcal{G}_L}$, we estimate it as:
{\small
\begin{align}
 \frac{\sum_{m=1}^M 1\qty{ \qty{\hat{{G}}^{(m)},\hat{{G}}^{'(m)}} \subseteq \qty{{G}_1,\ldots,  {G}_L }, p^{(m)}\leq \alpha}}{\sum_{m=1}^M 1\qty{ \qty{\hat{{G}}^{(m)}_1,\hat{{G}}^{'(m)}_2} \subseteq \qty{{G}_1,\ldots,  {G}_L } }},  \label{eq:conditional_power} 
\end{align}
}
where $p^{(m)}$ and $\hat{{G}}^{(m)}, \hat{{G}}^{'(m)}$ denote the $p$-value and estimated clusters under consideration for the $m$th simulated data set. Because the quantity in \eqref{eq:conditional_power} conditions on the event that $\hat{{G}}_1$ and $\hat{{G}}_2$ are true clusters, we also estimate how often that event occurs, i.e., the \emph{detection probability}:
{\small\begin{align}
\label{eq:detect_p} 
&\sum_{m=1}^M 1\qty{ \qty{\hat{{G}}^{(m)}_1,\hat{{G}}^{(m)}_2} \subseteq \qty{{G}_1,\ldots,  {G}_L } }/M. 
\end{align}}

Figures~\ref{fig:sim_power} displays the conditional power \eqref{eq:conditional_power} for the tests based on $p_{j, \text{k-means}}$, $p_{j, \text{average}}$, $p_{j, \text{centroid}}$, or $p_{j, \text{single}}$. In cases where simulations did not recover the true clusters, we've conventionally set the conditional power to zero.  Under model \eqref{eq:mod} with $\mu$ defined in \eqref{eq:power_model}, the conditional power increases as a function of the difference in feature means $\delta$ across all proposed $p$-values and feature correlation $\rho$. For a given $q$, a larger value of $\rho$ leads to lower conditional power, especially for the test based on $p_{j, \text{k-means}}$. Moreover, for a given value of $\delta$ and $q$, the ordering of power for different tests depends on the correlation between features: with independent features (left column of Figure~\ref{eq:conditional_power}),  the test based $p_{j, \text{average}}$ and  $p_{j, \text{k-means}}$. By contrast, when features are highly correlated (right column of Figure~\ref{eq:conditional_power}), the tests based on $p_{j, \text{single}}$ and $p_{j, \text{k-means}}$ yield the highest and the lowest power, respectively. 

The observed trends are congruent with the anticipated behaviour of individual clustering algorithms: For instance, $k$-means clustering, which uses within-cluster-sum-of-squares, tends to underperform when features are highly correlated. By contrast, single linkage hierarchical clustering, making use of the minimal distance between clusters, thrives in settings with high signal-to-noise ratio and high feature correlation. Figure~\ref{fig:detection_prob} displays the relative performance of cluster recovery, characterized using detection probability \eqref{eq:detect_p}.
 
\begin{figure}[htbp!]
  \centering  \includegraphics[width=\linewidth]{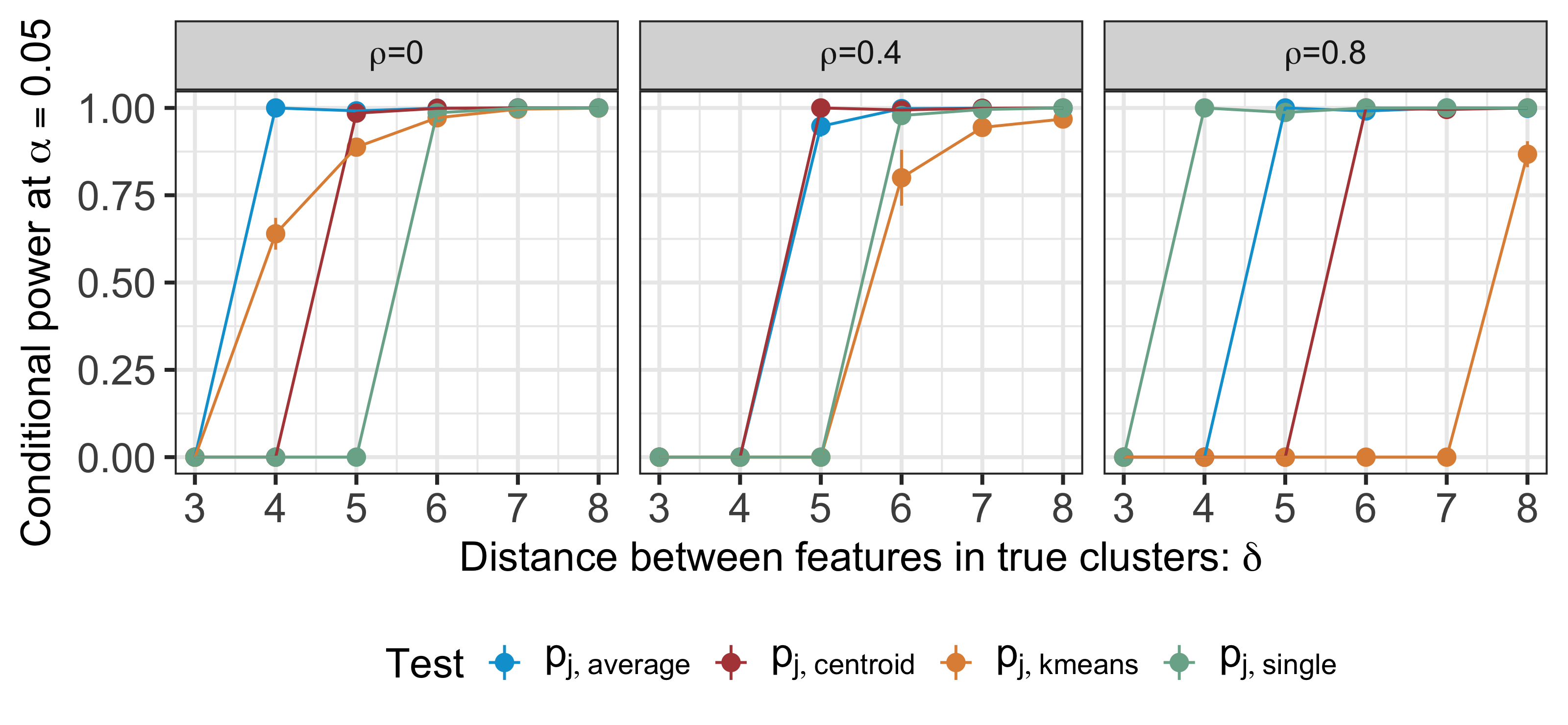}
\caption{The conditional power \eqref{eq:conditional_power} at $\alpha=0.05$ for the tests based on $p_{j, \text{k-means}}$, $p_{j, \text{average}}$, $p_{j, \text{centroid}}$, and $p_{j, \text{single}}$, under model \eqref{eq:mod} with $\mu$ defined in \eqref{eq:power_model}; $q=10$; and $\rho=0,0.4,0.8$.}
\label{fig:sim_power}
\end{figure}

\begin{figure}[htbp!]
  \centering
\includegraphics[width=1\linewidth]{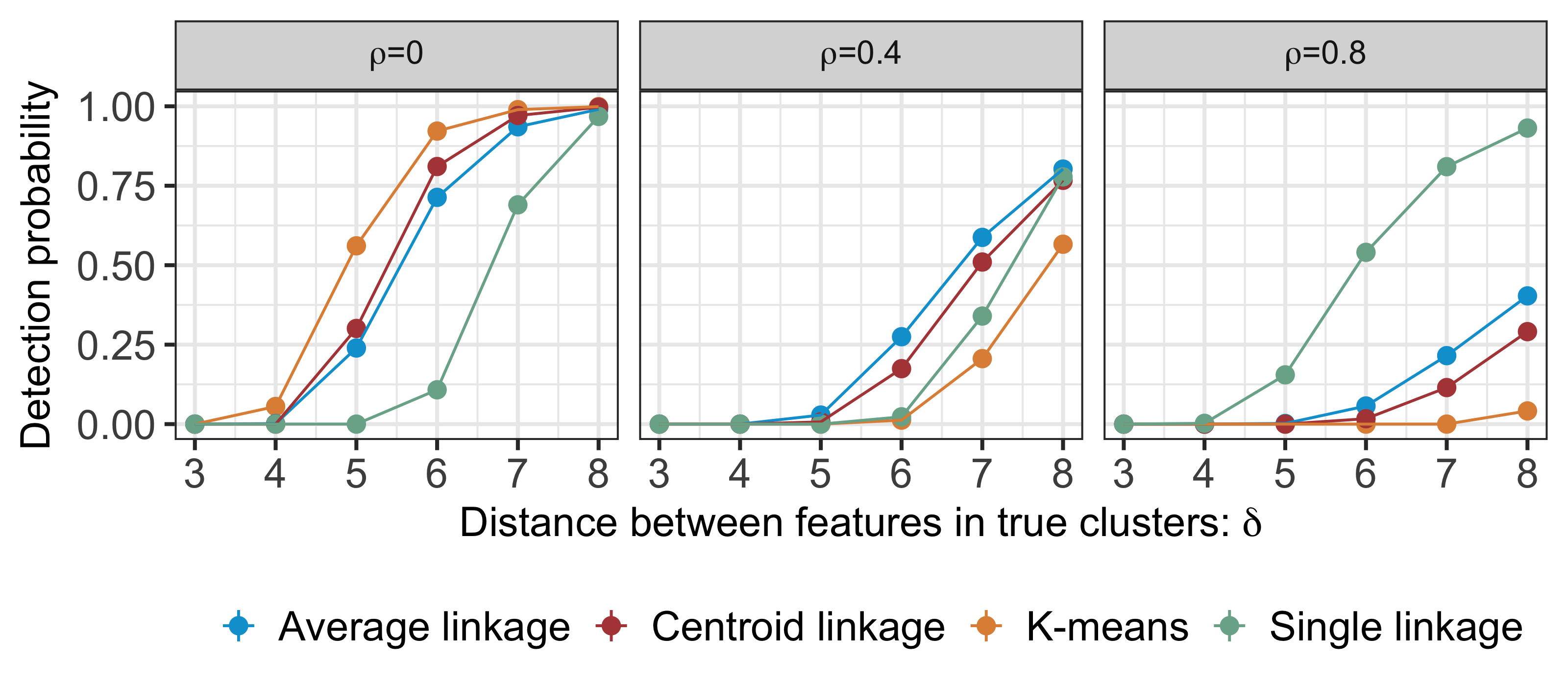}
\caption{The detection probability \eqref{eq:detect_p} at $\alpha=0.05$ of the underlying clustering methods for $p_{j, \text{k-means}}$, $p_{j, \text{average}}$, $p_{j, \text{centroid}}$, and $p_{j, \text{single}}$, under model \eqref{eq:mod} with $\mu$ defined in \eqref{eq:power_model}; $q=10$; and $\rho=0,0.4,0.8$.}
\label{fig:detection_prob}
\end{figure}
\section{Applications to scRNA-seq data}
\label{sec:application}

In this section, we apply the proposed $p$-values to single-cell RNA-sequencing data collected by the Tabula Sapiens Consortium~\citep{tabula2022tabula}, which measures messenger RNA expression levels in each of 500,000 cells from 24 different tissues and organs. These data have enabled new insights into the distinct cell types within the human organism and created a detailed molecular definition of these cell types. To reveal biological insights on how gene expression levels change across cell types, biologists typically perform clustering on the cells, and then perform a differential expression analysis, i.e., they test for a difference in gene expression between two clusters~\citep{aizarani2019human,grun2015single}. In this approach, ignoring the fact that the clusters were estimated from the same data used for testing, e.g., applying a two-sample $Z$-test instead of a selective test in the differential expression analysis, will inflate the Type I error rate. 

\begin{figure}[htbp!]
  \centering
\includegraphics[width=1\linewidth]{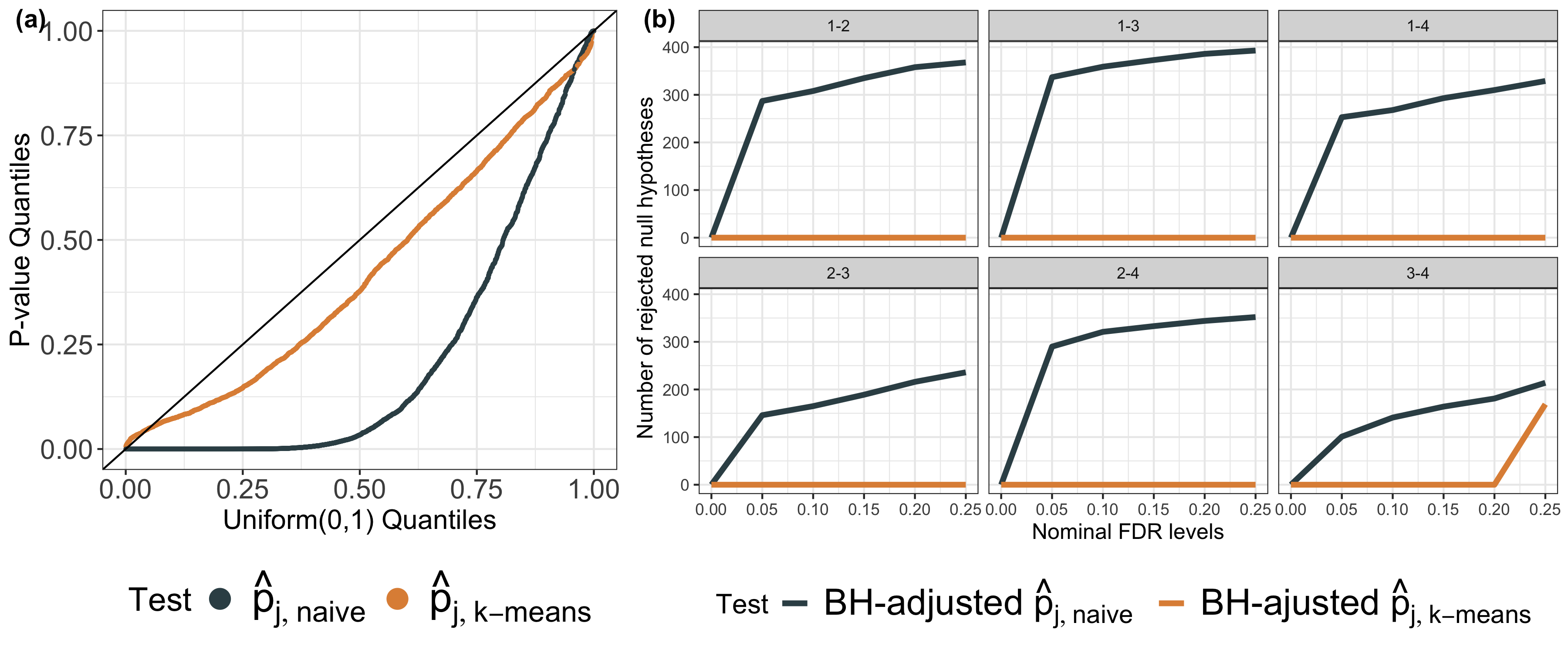}
\caption{\textit{(a): }Quantile-quantile plot of the $p$-values $\hat{p}_{j, \text{naive}}$ and $\hat{p}_{j, \text{k-means}}$ aggregated over features $j=1,\ldots, 500$ and all pairs of estimated clusters on the ``no cluster'' data set. \textit{(b): }Number of rejected null hypotheses at different nominal FDR levels using BH-procedure-adjusted $p$-values from \textit{(a)}.}
\label{fig:real_data_single_cell}
\end{figure}

One unique feature of the Consortium data set is that experts annotated cell types consistently across the different tissues. We will make use of the labelled cell types as the ``ground truth'' and use this information to demonstrate that our proposed $p$-values in Section~\ref{sec:select} yield biologically reasonable results. As per standard pre-processing techniques~\citep{Duo2018-ap}, we first excluded cells with low numbers or total counts of expressed genes, as well as cells in which a large percentage of the expressed genes are mitochondrial. We then normalized the transcripts for each cell by the total sum of counts in that cell, followed by a $\log_2(x+1)$ transformation. 

\begin{figure}[htbp!]
  \centering
\includegraphics[width=\linewidth]{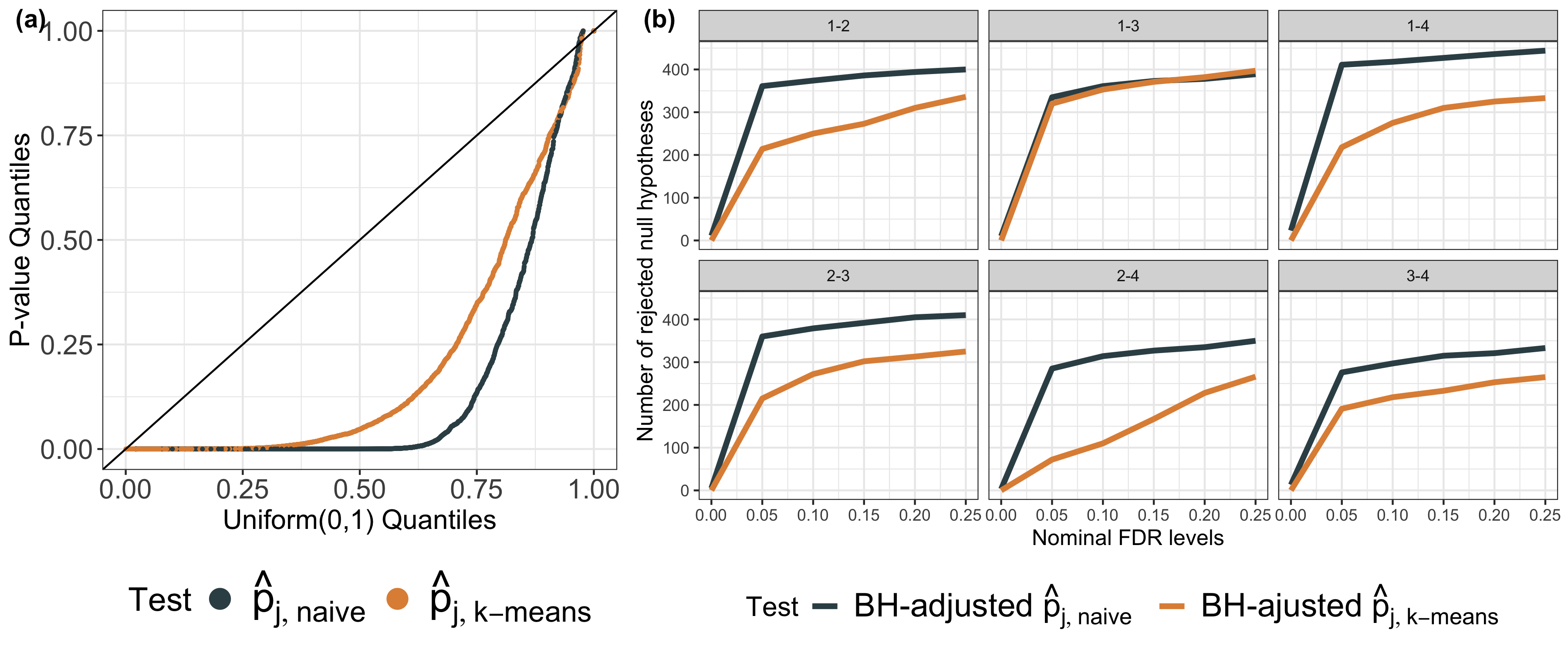}
\caption{\textit{(a): }Quantile-quantile plot of the $p$-values $\hat{p}_{j, \text{naive}}$ and $\hat{p}_{j, \text{k-means}}$ aggregated over features $j=1,\ldots, 500$ and all pairs of estimated clusters on the ``cluster'' data set. \textit{(b): }Number of rejected null hypotheses at different nominal FDR levels using BH-adjusted $p$-values from \textit{(a)}.}
\label{fig:real_data_single_cell_H1}
\end{figure}

We applied the pre-processing pipeline separately on two sets of cells collected from the same donor: the CD4-positive, alpha-beta T cells, and a combined sample of four cellular types from an identical donor, namely memory B cells, natural killer cells, macrophages, and monocytes. We considered only the subset of 500 genes with the largest sample variance in expression levels post-normalization.

To investigate the selective type I error rate in the absence of true clusters, we first consider a ``no cluster'' data set consisting of only CD4-positive, alpha-beta T cells after pre-processing (thus, $n=833$ and $q = 500$). We applied $k$-means clustering with $K=4$ to obtain four estimated clusters. For each pair of estimated clusters and each feature $j=1,\ldots, 500$, we computed the $p$-values $\hat{p}_{j, \text{naive}}$ and $\hat{p}_{j, \text{k-means}}$ (where the $\hat{p}$ emphasizes that we used the sample covariance matrix as an estimate of $\Sigma$ in \eqref{eq:mod}). The quantile-quantile plot of the resulting $p$-values is displayed in Figure~\ref{fig:real_data_single_cell}(a). We display the number of rejected hypotheses after FDR correction using the BH procedure~\citep{Benjamini1995-lo} against the nominal FDR level in Figure~\ref{fig:real_data_single_cell}(b). In this data set, the naive $p$-values are extremely small for all pairs of estimated clusters, leading to hundreds of rejected null hypotheses of equal means, while our proposed $p$-values are quite large and lead to virtually no rejections after FDR correction. In particular, at $\alpha=0.05$ and FDR level of $0.20$, the test based on  $\hat{p}_{j, \text{naive}}$ would conclude that more than 60\% genes are ``differentially expressed'', whereas our approach would suggest that expression levels across clusters are the same for most genes. Because this ``no cluster'' data set consists only of a single type of expert-annotated cell from a single donor, we believe the conclusion based on $\hat{p}_{j, \text{k-means}}$ aligns better with the underlying biology.

Next, we turn our attention to the ``cluster'' data set, which encompasses memory B cells, natural killer cells, macrophages, and monocytes. We applied $k$-means clustering to obtain four clusters, and subsequently estimated a covariance matrix based on the residuals from the $k$-means fit. Notably, the clusters derived in this manner align closely with the four distinct cell types, with an adjusted Rand Index between the cell type annotations and estimated cluster memberships exceeding 0.6. We then computed the $p$-values $\hat{p}_{j, \text{naive}}$ and $\hat{p}_{j, \text{k-means}}$ across all features and for each pair of the estimated clusters. The quantiles of these $p$-values, as well as the number of null hypotheses rejected following FDR adjustment, are depicted in panels (a) and (b) of Figure~\ref{fig:real_data_single_cell_H1}, respectively. Notably, both sets of $p$-values on this data set are quite small, and the BH procedure results in a comparable count of rejections for both sets of $p$-values. Given that the clusters in this context largely correspond to distinct cell types, our results suggest that the test employing our proposed $p$-value has reasonable power to reject the null hypothesis when it does not hold.

\section{Discussion}
\label{sec:disucssion}

In this work, we proposed a test for a difference in means for a single feature between two clusters estimated from hierarchical or $k$-means clustering, under \eqref{eq:mod}. Here, we outline several future research directions.

The $p$-values introduced in Section~\ref{sec:select} can be extended to test for a difference in means between \emph{groups} of features. For instance, if we want to test for equality in means for all features $j\in J$ between two estimated clusters, i.e., $
\hat H_{0J}: \bigcap_{j\in J}\big\{\bar{\mu}_{\hat G j } =  \bar{\mu}_{\hat G' j }\big\}$. Following the argument in this paper, the following $p$-value 
{\footnotesize$\mathbb{P}_{\hat H_{0J}} \Big (  \big \Vert[X^T \hat \nu]_{J} \big\Vert_2 \big | \geq \big \Vert[x^T \hat \nu]_{J} \big\Vert_2 \big | ~ \Bigg |~ \mathcal{C}(X) = \mathcal{C}(x), $} {\footnotesize $ X - 
\frac{\hat \nu \Sigma_{J}^T [X^T \hat \nu]_J }{\|\hat \nu\|_2^2 \Sigma_{JJ}} = x - 
\frac{\hat \nu \Sigma_{J}^T [x^T \hat \nu]_J }{\|\hat \nu\|_2^2 \Sigma_{JJ}} \Big )$} controls the selective Type I error rate under \eqref{eq:mod} and can be efficiently computed. Here,  $[x^T \hat \nu]_{J} \in \mathbb{R}^{|J|}$ represents the vector subset with indices in $J$;  $\Sigma_{J} \in \mathbb{R}^{q\times |J|}$ is the submatrix in $\Sigma$ with columns in $J$; and $\Sigma_{JJ} \in \mathbb{R}^{|J|\times |J|}$ is the submatrix in $\Sigma$ with row and column indices in $J$.

Furthermore, our $p$-values can be used to derive selective confidence intervals for a difference in means for feature $j$~\citep{Lee2016-te,fithian2014optimal}. Under our setup, computing the confidence intervals amounts to a root-finding problem, which can be efficiently solved using bisection~\citep{chen2023more,Chen2020-rh}. This extension would enhance data uncertainty evaluation: for instance, both $p$-values and confidence intervals on differing gene expression profiles are used to guide downstream scientific inquiries.







\section*{Acknowledgments}

We acknowledge funding from the following sources: Natural Sciences and Engineering Research Council of Canada Discovery Grant to LG and Stanford Data Science Fellowship to YC. This work was also partially supported by National Institutes of Health grants (R01 EB026908, R01 GM123993, and R01 DA047869).

\bibliographystyle{apalike}
\bibliography{testing_sf}

\begin{thebibliography}{}

\bibitem[Aizarani et~al., 2019]{aizarani2019human}
Aizarani, N., Saviano, A., Mailly, L., Durand, S., Herman, J.~S., Pessaux, P., Baumert, T.~F., and Gr{\"u}n, D. (2019).
\newblock A human liver cell atlas reveals heterogeneity and epithelial progenitors.
\newblock {\em Nature}, 572(7768):199--204.

\bibitem[Aloise et~al., 2009]{Aloise2009-on}
Aloise, D., Deshpande, A., Hansen, P., and Popat, P. (2009).
\newblock {NP-hardness} of {Euclidean} sum-of-squares clustering.
\newblock {\em Machine Learning}, 75(2):245--248.

\bibitem[Benjamini and Hochberg, 1995]{Benjamini1995-lo}
Benjamini, Y. and Hochberg, Y. (1995).
\newblock Controlling the false discovery rate: A practical and powerful approach to multiple testing.
\newblock {\em Journal of the Royal Statistical Society. Series B, Statistical methodology}, 57(1):289--300.

\bibitem[Chen and Bien, 2020]{Chen2020-rh}
Chen, S. and Bien, J. (2020).
\newblock Valid inference corrected for outlier removal.
\newblock {\em Journal of Computational and Graphical Statistics}, 29(2):323--334.

\bibitem[Chen et~al., 2023]{chen2023more}
Chen, Y., Jewell, S., and Witten, D. (2023).
\newblock More powerful selective inference for the graph fused lasso.
\newblock {\em Journal of Computational and Graphical Statistics}, 32(2):577--587.

\bibitem[Chen and Witten, 2022]{chen2022selective}
Chen, Y.~T. and Witten, D.~M. (2022).
\newblock Selective inference for k-means clustering.
\newblock {\em To appear in Journal of Machine Learning Research}.

\bibitem[Consortium et~al., 2022]{tabula2022tabula}
Consortium, T.~S., Jones, R.~C., Karkanias, J., Krasnow, M.~A., Pisco, A.~O., Quake, S.~R., Salzman, J., Yosef, N., Bulthaup, B., Brown, P., et~al. (2022).
\newblock The tabula sapiens: A multiple-organ, single-cell transcriptomic atlas of humans.
\newblock {\em Science}, 376(6594):eabl4896.

\bibitem[Du{\`o} et~al., 2018]{Duo2018-ap}
Du{\`o}, A., Robinson, M.~D., and Soneson, C. (2018).
\newblock A systematic performance evaluation of clustering methods for single-cell {RNA-seq} data.
\newblock {\em F1000Research}, 7:1141.

\bibitem[Fithian et~al., 2014]{fithian2014optimal}
Fithian, W., Sun, D., and Taylor, J. (2014).
\newblock Optimal inference after model selection.
\newblock {\em arXiv preprint arXiv:1410.2597}.

\bibitem[Gao et~al., 2022]{gao2022selective}
Gao, L.~L., Bien, J., and Witten, D. (2022).
\newblock Selective inference for hierarchical clustering.
\newblock {\em Journal of the American Statistical Association}, pages 1--11.

\bibitem[Gr{\"u}n et~al., 2015]{grun2015single}
Gr{\"u}n, D., Lyubimova, A., Kester, L., Wiebrands, K., Basak, O., Sasaki, N., Clevers, H., and Van~Oudenaarden, A. (2015).
\newblock Single-cell messenger rna sequencing reveals rare intestinal cell types.
\newblock {\em Nature}, 525(7568):251--255.

\bibitem[{Hastie} et~al., 2001]{Hastie2001-mr}
{Hastie}, {Trevor.}, {Hastie}, {Trevor.}, {Tibshirani}, {Robert.}, {Friedman}, and H., J. (2001).
\newblock {\em The Elements of Statistical Learning : data mining, inference, and prediction}.
\newblock Springer, New York.

\bibitem[Hivert et~al., 2022]{hivert2022post}
Hivert, B., Agniel, D., Thi{\'e}baut, R., and Hejblum, B.~P. (2022).
\newblock Post-clustering difference testing: valid inference and practical considerations.
\newblock {\em arXiv preprint arXiv:2210.13172}.

\bibitem[Hyun et~al., 2021]{Hyun2021-xm}
Hyun, S., Lin, K.~Z., G'Sell, M., and Tibshirani, R.~J. (2021).
\newblock Post-selection inference for changepoint detection algorithms with application to copy number variation data.
\newblock {\em Biometrics}.

\bibitem[Jaeger and Banks, 2022]{jaeger2022cluster}
Jaeger, A. and Banks, D. (2022).
\newblock Cluster analysis: A modern statistical review.
\newblock {\em Wiley Interdisciplinary Reviews: Computational Statistics}, page e1597.

\bibitem[Jewell et~al., 2022]{Jewell2022-fn}
Jewell, S., Fearnhead, P., and Witten, D. (2022).
\newblock Testing for a change in mean after changepoint detection.
\newblock {\em Journal of the Royal Statistical Society Series B: Statistical Methodology}, 84(4):1082--1104.

\bibitem[Lee et~al., 2016]{Lee2016-te}
Lee, J.~D., Sun, D.~L., Sun, Y., and Taylor, J.~E. (2016).
\newblock Exact post-selection inference, with application to the lasso.
\newblock {\em The Annals of Statistics}, 44(3):907--927.

\bibitem[Leisch et~al., 2018]{leisch2018market}
Leisch, F., Dolnicar, S., and Gr{\"u}n, B. (2018).
\newblock {\em Market segmentation analysis: Understanding it, doing it, and making it useful}.
\newblock Springer.

\bibitem[Lloyd, 1982]{Lloyd2006-io}
Lloyd, S. (1982).
\newblock Least squares quantization in {PCM}.
\newblock {\em IEEE Trans. Inf. Theory}, 28(2):129--137.

\bibitem[Neufeld et~al., 2022]{neufeldinference2022}
Neufeld, A., Gao, L.~L., Popp, J., Battle, A., and Witten, D. (2022).
\newblock Inference after latent variable estimation for single-cell {RNA} sequencing data.
\newblock {\em Biostatistics}, pages 1--18.

\bibitem[Wood, 2017]{wood2017}
Wood, S. (2017).
\newblock {\em Generalized Additive Models: An Introduction with R}.
\newblock Chapman and Hall/CRC, 2 edition.

\end{thebibliography}
\section*{Appendix}
\appendix

\newcommand{\indep}{\perp \!\!\! \perp}

\section{Proof of Theorem 1}
\label{sec:proof_theorem_1}
We now prove Theorem 1. Observe that 
\begin{align} 
X 
&= X - \frac{\hat \nu}{\|\hat \nu\|_2^2} \left ( \frac{\Sigma_j}{\Sigma_{jj}}\right )^T [X^T \hat \nu]_j  + \frac{\hat \nu}{\|\hat \nu\|_2^2} \left ( \frac{\Sigma_j}{\Sigma_{jj}}\right )^T [X^T \hat \nu]_j \nonumber \\ 
&= U(X) + \frac{\hat \nu}{\|\hat \nu\|_2^2} \left ( \frac{\Sigma_j}{\Sigma_{jj}}\right )^T [X^T \hat \nu]_j, \label{eq:decomp-appendix}
\end{align} 
where $U(X)$ is defined in equation (13) of the main text. 

Substituting \eqref{eq:decomp-appendix} into the definition of $p_{j, \text{selective}}$ in equation \eqref{eq:selective} of the main text yields: 
\begin{align} 
p_{j, \text{selective}}  &= \mathbb{P}_{\hat H_{0j}} \Big (  \big |[X^T \hat \nu]_j \big | \geq \big |[x^T \hat \nu]_j  \big | ~ \Big |~ \mathcal{C} \left (U(x) +   \frac{\hat \nu}{\|\hat \nu\|_2^2} \left ( \frac{\Sigma_j}{\Sigma_{jj}}\right )^T [X^T \hat \nu]_j \right ) = \mathcal{C}(x),~ U(X) = U(x) \Big ) \nonumber \\ 
&=\mathbb{P}_{\hat H_{0j}} \Big (  \big |[X^T \hat \nu]_j \big | \geq \big |[x^T \hat \nu]_j  \big | ~ \Big |~ \mathcal{C} \left (x'([X^T \hat \nu]_j, j) \right ) = \mathcal{C}(x), ~U(X) = U(x) \Big ), \label{eq:sub} 
\end{align} 
where the second equality in \eqref{eq:sub} follows from the definition of $x'(\phi, j)$ in Theorem 1.  For any $\phi \in \mathbb{R}$, we can rewrite $x'(\phi, j)$ as:
\begin{align*} 
&x + (\phi - ( \bar{x}_{\hat{G}j} - \bar{x}_{\hat{G'}j})) \left ( \frac{ \hat\nu }{ \|\hat\nu\|_2^2 } \right ) \left ( \frac{\Sigma_j}{\Sigma_{jj}} \right )^T \\ 
&= x + (\phi - [x^T \hat \nu]_j]) \left ( \frac{ \hat\nu }{ \|\hat\nu\|_2^2 } \right ) \left ( \frac{\Sigma_j}{\Sigma_{jj}} \right )^T \\ 
&= U(x) + \left ( \frac{ \hat\nu }{ \|\hat\nu\|_2^2 } \right ) \left ( \frac{\Sigma_j}{\Sigma_{jj}} \right )^T \phi.
\end{align*} 

To simplify \eqref{eq:sub}, we will show that 
\begin{align} 
U(X) \indep [\hat X^T \hat \nu]_j. \label{eq:ind}
\end{align} 

First, observe that 
\begin{align*} 
U(X) &= X - \frac{\hat \nu}{\|\hat \nu\|_2^2} \left ( \frac{\Sigma_j}{\Sigma_{jj}}\right )^T [X^T \hat \nu]_j \\ 
&= X - \frac{\hat \nu \hat \nu^T}{\|\hat \nu\|_2^2} X + \frac{\hat \nu \hat \nu^T}{\|\hat \nu\|_2^2} X - \frac{\hat \nu}{\|\hat \nu\|_2^2} \left ( \frac{\Sigma_j}{\Sigma_{jj}}\right )^T [X^T \hat \nu]_j \\ 
&=  \left (I_n - \frac{\hat \nu \hat \nu^T}{\|\hat \nu\|_2^2} \right ) X + \frac{\hat \nu }{\|\hat \nu\|_2^2} \left ( X^T \hat \nu - \frac{\Sigma_j}{\Sigma_{jj}} [X^T \hat \nu]_j \right )^T. 
\end{align*} 
Thus, to show \eqref{eq:ind}, it suffices to show that the two terms in the sum are independent of $ [\hat X^T \hat \nu]_j$; i.e., 
\begin{align} 
\left (I_n - \frac{\hat \nu \hat \nu^T}{\|\hat \nu\|_2^2} \right ) X \indep \hat X^T \hat \nu \label{eq:orthog}
\end{align} 
and 
\begin{align} 
\left ( X^T \hat \nu - \frac{\Sigma_j}{\Sigma_{jj}} [X^T \hat \nu]_j \right ) \indep [X^T \hat \nu]_j. \label{eq:therest}
\end{align} 

We will start by showing \eqref{eq:orthog}. Observe that $I_n - \frac{\hat \nu \hat \nu^T}{\|\hat \nu\|_2^2}$ is the orthogonal projection matrix onto the subspace orthogonal to $\hat \nu$. Thus, $\left (I_n - \frac{\hat \nu \hat \nu^T}{\|\hat \nu\|_2^2} \right ) \hat \nu = 0$, and  \eqref{eq:orthog} follows from properties of the multivariate normal distribution. 

To show \eqref{eq:therest}, we will need to state and prove an intermediate result. 

\begin{lemma} 
\label{lem:obscure}
Suppose that $\left [ \begin{matrix} z_1 \\ z_2 \end{matrix} \right ] \sim N_q \left (\left [ \begin{matrix} \mu_1 \\ \mu_2 \end{matrix} \right ] , \left[ \begin{matrix}  \Sigma_{11} & \Sigma_{12} \\ \Sigma_{21} & \Sigma_{22} \end{matrix} \right ]  \right )$, where $z_1$ is a $r$-vector and $z_2$ is a $(q-r)$-vector, and $\Sigma_{11}$ is invertible. Then, 
$z_1$ and $z_2 -  \Sigma_{21} \Sigma_{11}^{-1} z_1$ are independent.
\end{lemma}   
\begin{proof} 
Observe that 
$$\left [  \begin{matrix} z_1 \\ z_2 -   \Sigma_{21}   \Sigma_{11}^{-1} z_1 \end{matrix} \right ] = \left [ \begin{matrix} I_r &    0_{r \times (q-r)} \\ -  \Sigma_{21}     \Sigma_{11}^{-1} & I_{q-r} \end{matrix} \right ]  \left [ \begin{matrix} z_1 \\ z_2 \end{matrix} \right ]. $$
Therefore, for $B =  \left [ \begin{matrix} I_r &    0_{r \times (q-r)} \\ -  \Sigma_{21}     \Sigma_{11}^{-1} & I_{q-r} \end{matrix} \right ]$, we have 
$$ \left [  \begin{matrix} z_1 \\ z_2 -   \Sigma_{21}   \Sigma_{11}^{-1} z_1 \end{matrix} \right ] \sim N_q\left ( B \left [ \begin{matrix} \mu_1 \\ \mu_2 \end{matrix}  \right ] ,   B\left[ \begin{matrix}   \Sigma_{11} &   \Sigma_{12} \\   \Sigma_{21} &   \Sigma_{22} \end{matrix} \right ]B^T \right ).$$ 
Algebra yields 
$$  B \left[ \begin{matrix}   \Sigma_{11} &   \Sigma_{12} \\   \Sigma_{21} &   \Sigma_{22} \end{matrix} \right ]B^T = \left [ \begin{matrix}   \Sigma_{11} &   0_{r \times (q-r)} \\   0_{(q-r) \times r} &   \Sigma_{22} -   \Sigma_{21}   \Sigma_{11}^{-1}   \Sigma_{12}\end{matrix} \right ].$$
Thus, it follows from properties of multivariate normal distributions that $z_1$ and $z_2 -   \Sigma_{21}  \Sigma_{11}^{-1} z_1$ are independent. 

\end{proof} 

Since $X^T \hat \nu \sim N_q(\mu^T \hat \nu, \|\hat \nu\|_2^2 \Sigma)$, it follows from Lemma \ref{lem:obscure} that 
\begin{align} 
[X^T \hat \nu]_j \indep [X^T \hat \nu]_{-j} - \frac{\Sigma_{j, -j}^T}{\Sigma_{jj}} [X^T \hat \nu]_j, \label{eq:applylem}
\end{align} 
where $[X^T \hat \nu]_{-j}$ denotes the result of removing the $j$th component from $X^T \hat \nu \in \mathbb{R}^q$, and $\Sigma_{j, -j}$ denotes the result of removing the $j$th component from $\Sigma_j \in \mathbb{R}^q$. Finally, to complete the proof of \eqref{eq:therest}, note that by algebra, the $j$th entry of the vector $X^T \hat \nu - \frac{\Sigma_j}{\Sigma_{jj}} [X^T \hat \nu]_j$ is always 0. This implies that any random variable that is independent of the $q-1$ subvector $[X^T \hat \nu]_{-j} - \frac{\Sigma_{j, -j}^T}{\Sigma_{jj}} [X^T \hat \nu]_j$ in \eqref{eq:applylem} is also independent of the full vector; this completes the proof.

We have now shown \eqref{eq:ind}. Thus, we can apply \eqref{eq:ind} to simplify \eqref{eq:sub} as: 
\begin{align} 
p_{j, \text{selective}}
&=\mathbb{P}_{\hat H_{0j}} \Big (  \big |[X^T \hat \nu]_j \big | \geq \big |[x^T \hat \nu]_j  \big | ~ \Big |~ \mathcal{C} \left (x'([X^T \hat \nu]_j, j) \right ) = \mathcal{C}(x) \Big ) \nonumber \\ 
&= \mathbb{P}_{\hat H_{0j}} \Big (  \big |[X^T \hat \nu]_j \big | \geq \big |[x^T \hat \nu]_j  \big | ~ \Big |~ [X^T \hat \nu]_j \in \hat S_j \Big ), \label{eq:trunc}
\end{align}
where the second equality follows from the definition of $\hat S_j$ in equation \eqref{eq:defS} of the main text as $\left \{ \phi \in \mathbb{R}: C(x) = \mathcal{C}\left(x'(\phi,j)\right ) \right \}$. 

Now, under $\hat H_{0j}: [\mu^T \hat \nu]_j = 0$, $[X^T \hat \nu]_j \sim N(0, \|\hat \nu\|_2^2 \Sigma_{jj})$. Thus, defining $\phi \sim N(0, \|\hat \nu\|_2^2 \Sigma_{jj})$, we can rewrite \eqref{eq:trunc} as: 
\begin{align} 
p_{j, selective} &= \mathbb{P}\Big ( \big | \phi \big | \geq \big |[x^T \hat \nu]_j  \big | ~ \Big | ~ \phi \in \hat S_j \Big ) \nonumber \\ 
&= \mathbb{P}\Big( \phi \geq \big |[x^T \hat \nu]_j  \big | ~ \Big | ~ \phi \in \hat S_j \Big ) +  \mathbb{P}\Big( \phi \leq -\big |[x^T \hat \nu]_j  \big | ~ \Big | ~ \phi \in \hat S_j \Big ) \nonumber \\ 
&= 1 - \mathbb{P}\Big( \phi \leq \big |[x^T \hat \nu]_j  \big | ~ \Big | ~ \phi \in \hat S_j \Big ) +  \mathbb{P}\Big( \phi \leq -\big |[x^T \hat \nu]_j  \big | ~ \Big | ~ \phi \in \hat S_j \Big )  \nonumber \\ 
&= 1 - \mathbb{F} \left (\big | [\hat \nu^T x]_j \big | ; 0, \Sigma_{jj}\|\hat \nu\|_2^2; \mathcal{\hat S}_j \right ) 
+ \mathbb{F}\left (-\big | [\hat \nu^T x]_j \big | ; 0, \Sigma_{jj} \|\hat \nu\|_2^2;  \mathcal{\hat S}_j \right ) \label{eq:trunc-norm}, 
\end{align} 
where $\mathbb{F}(t; \mu, \sigma, \mathcal{S})$ denotes the cumulative distribution function (CDF) of a $N(\mu, \sigma^2)$ random variable truncated to the set $\mathcal{S}$. This is equation \eqref{eq:selective_cdf} in the main text. 

We will now prove that rejecting $\hat H_{0j}$ based on $p_{j,\text{selective}} \leq \alpha$ controls the selective type I error rate in the sense of \eqref{eq:selective_type_1}. First, recall that we decided to test $\hat H_{0j}: [\mu^T \hat \nu]_j = 0$ because $\mathcal{C}(X) = \mathcal{C}(x)$. Thus, we need to show: 
\begin{align} 
\mathbb{P}_{\hat H_{0j}} \left (\text{Reject } \hat H_{0j} \text{ at level } \alpha \mid \mathcal{C}(X) = \mathcal{C}(x) \right ) \leq \alpha, ~~~ \text{for all } \alpha \in (0, 1). \label{eq:type1-app}
\end{align} 

Define the function 
\begin{align} 
p(t) = \mathbb{P}\Big ( \big | \phi \big | \geq \big |t \big | ~ \Big | ~ \phi \in \hat S_j \Big ) \label{eq:pval-func}
\end{align} 
Then, observe that the following holds for any $\alpha \in (0,1)$: 
\begin{align} 
&\mathbb{P}_{\hat H_{0j}} \left (\text{Reject } \hat H_{0j} \text{ at level } \alpha \mid \mathcal{C}(X) = \mathcal{C}(x), U(X) = U(x) \right ) \nonumber \\
&= \mathbb{P}( p(\phi) \leq \alpha \mid \phi \in \hat S_j) \nonumber \\ 
&= \alpha, \label{eq:pit}
\end{align} 
where the first equality follows from \eqref{eq:trunc-norm}, and the second equality follows from the probability integral transform theorem. 
Therefore, we have that
\begin{align*} 
&\mathbb{P}_{\hat H_{0j}} \left (\text{Reject } \hat H_{0j} \text{ at level } \alpha \mid \mathcal{C}(X) = \mathcal{C}(x) \right ) \\
&= \mathbb{E}_{\hat H_{0j}} \left [ \mathbb{P}_{\hat H_{0j}} \left (\text{Reject } \hat H_{0j} \text{ at level } \alpha \mid \mathcal{C}(X) = \mathcal{C}(x), U(X) = U(x) \right )  \mid  \mathcal{C}(X) = \mathcal{C}(x) \right ] \\ 
&= \mathbb{E}_{\hat H_{0j}} [\alpha \mid \mathcal{C}(X) = \mathcal{C}(x) ] \\ 
&= \alpha,
\end{align*}
where the first equality follows from the tower property of conditional expectation, and the second equality follows from \eqref{eq:pit}. That is, equation \eqref{eq:type1-app} holds with equality.  

\section{Additional details for the extensions to $k$-means clustering in Section 4.3}

In this section, we provide more details on the extensions of proposals in Section 3 of the main text to $k$-means clustering.

\subsection{A brief overview of the $k$-means algorithm}

We briefly review the $k$-means clustering algorithm in this section. For a set of samples $x_1,\ldots,x_n \in \mathbb{R}^q$ and a positive integer $K$, $k$-means clustering partitions the $n$ samples into disjoint subsets $\hat{\mathcal{C}}_1,\ldots,\hat{\mathcal{C}}_K$  by solving the following optimization problem~\citep{Aloise2009-on,Lloyd2006-io}:
\begin{align}
\begin{split}
\label{eq:k_means_objective}
  &\underset{\mathcal{C}_1,\ldots,\mathcal{C}_K}{\text{minimize}}\;\Bigg\{ \sum_{k=1}^K \sum_{i \in \mathcal{C}_k} \bigg\Vert x_i -  \sum_{i \in \mathcal{C}_k} x_i/|\mathcal{C}_k| \bigg\Vert_2^2 \Bigg\} \\
  &\text{subject to} \;
 \bigcup_{k=1}^K \mathcal{C}_k = \{1,\ldots, n\},\mathcal{C}_k\cap \mathcal{C}_{k'} = \emptyset ,\forall k\neq k'.
\end{split}
\end{align} 
It is not typically possible to solve for the global optimum in \eqref{eq:k_means_objective}~\citep{Aloise2009-on}; one of the most popular approaches is Lloyd's algorithm~\citep{Lloyd2006-io}, given in Algorithm~\ref{algo:k_means_alt_min}. In Lloyd's algorithm, we first sample $K$ out of $n$ observations as initial centroids (step 1 in Algorithm~\ref{algo:k_means_alt_min}), followed by assigning each observation to its closest centroid (step 2). Next, we iterate between steps 1 and 2 until the cluster assignments stop changing, and the algorithm is guaranteed to converge to a local optimum~\citep{Hastie2001-mr}.

\begin{algorithm}
\DontPrintSemicolon 
\KwIn{Data $x_1,\ldots,x_n \in \mathbb{R}^q$, number of output clusters $K$, maximum iteration $T$, random seed $s$.}
\KwOut{Cluster assignments $(c_1^{(t)},\ldots, c_n^{(t)})$.}
1. Initialize the centroids $(m_1^{(1)},\ldots,m_K^{(1)})$ by sampling 
 $K$ observations from $x_1,\ldots, x_n$ without replacement, using the random seed $s$. \\ 
2. Compute assignments  $c_i^{(1)} \leftarrow \underset{1\leq k\leq K}{\text{argmin}} \left\Vert x_{i}-m_k^{(1)}\right\Vert_2^2,  \, i = 1,\ldots, n. $ \; \label{algo:step2}
3. Initialize $t=1$. \;
\While{$t\leq T$}{
  a. Update centroids:  $m_k^{(t+1)} \leftarrow (\sum_{i:c_i^{(t)}=k} x_i)/\sum_{i=1}^n \mathds{1}\{c_i^{(t)}=k\}, \,  k = 1,\ldots, K. $\;
  b. Update assignment:  $c_i^{(t+1)} \leftarrow \underset{1\leq k\leq K}{\text{argmin}} \left\Vert x_{i}-m_k^{(t+1)}\right\Vert_2^2,  \, i = 1,\ldots, n. $\;
  c. \textbf{if} $c_i^{(t+1)} = c_i^{(t)}$ \emph{for all} $1\leq i \leq n$ \\
  \quad\quad \textbf{break} \\
  \textbf{else} \\
  \quad\quad $t \leftarrow t+1. $\;
}
\Return{$(c_1^{(t)},\ldots, c_n^{(t)})$}.\;
\caption{Lloyd's algorithm for $k$-means clustering~\citep{Lloyd2006-io}}
\label{algo:k_means_alt_min}
\end{algorithm}

\subsection{Characterization and efficient computation of  $p_{j,\text{k-means}}$}
Theorem~\ref{thm:truncnorm_kmeans} below is a direct extension of Theorem 1 (stated in Section~\ref{section:truncated_pval} of the main text and proven in Appendix~\ref{sec:proof_theorem_1}) to the case of $k$-means clustering. Recalling $U(x) = x - \frac{\hat \nu \Sigma_{j}^T [x^T \hat \nu]_j }{\|\hat \nu\|_2^2 \Sigma_{jj}}$, we first rewrite $p_{j,\text{k-means}}$:
\begin{align*} 
p_{j,\text{k-means}} 
&= \mathbb{P}_{{\hat{H}_{0j}}} \Bigg (  \big |[X^T \hat\nu]_j \big | \geq \big |[X^T \hat\nu]_j  \big | ~ \Bigg | ~ U(X) = U(x),
\bigcap_{t=0}^{T}\bigcap_{i=1}^n\big\{ c_i^{(t)}(X) = c_i^{(t)}(x)\big\}  \Bigg  ).
\end{align*}
\begin{theorem} 
\label{thm:truncnorm_kmeans}
Suppose that $x$ is a realization from (1), i.e., each row $x_i$ is drawn independently from $N_q(\mu_i, \Sigma), \, i = 1, 2, \ldots, n,$ and $\mathbb{F}(t; \mu, \sigma, \mathcal{S})$ denotes the cumulative distribution function (CDF) of a $N(\mu, \sigma^2)$ random variable truncated to the set $\mathcal{S}$. If $\hat H_{0j}: \bar{\mu}_{\hat G j } =  \bar{\mu}_{\hat G' j }$  holds for a given $j \in \{1, 2, \ldots, q\}$, then we have that  
{\small\begin{align} 
p_{j,\text{k-means}} &= 1 - \mathbb{F} \left (\big | [\hat\nu^T x]_j \big | ; 0, \Sigma_{jj}\|\nu\|_2^2;\mathcal{S}_j\Big(x; \bigcap_{t=1}^T C^{(t)}(x)\Big) \right ) + \mathbb{F}\left (-\big | [\hat\nu^T x]_j \big | ; 0, \Sigma_{jj} \|\hat\nu\|_2^2; \mathcal{S}_j\Big(x; \bigcap_{t=1}^T C^{(t)}(x)\Big) \right ).
\end{align}}
Here, $C^{(t)}(x) = (c_{1}^{(t)},\ldots, c_{n}^{(1)})$ is the estimated cluster at the $t$th iteration of the $k$-means algorithm, and 
{\small
\begin{equation} 
\mathcal{S}_j\Big(x; \bigcap_{t=1}^T C^{(t)}(x)\Big)  = 
\left \{ \phi \in \mathbb{R}: \bigcap_{t=1}^T \left \{ {C}^{(t)}\left(x + (\phi - ( \bar{x}_{\hat{G}j} - \bar{x}_{\hat{G}'j})) \left ( \frac{ \hat\nu }{ \|\hat\nu\|_2^2 } \right ) \left ( \frac{\Sigma_j}{\Sigma_{jj}} \right )^T \right ) =  {C}^{(t)}(x)\right \} \right \}. \label{eq:defS_kmeans}
\end{equation} 
}
Furthermore, rejecting $\hat H_{0j}$ whenever $p_{\text{k-means},j} \leq \alpha$ controls the selective type I error rate at level $\alpha$. 
\end{theorem}  

\begin{proof}
    The proof of Theorem~\ref{thm:truncnorm_kmeans} follows directly from the proof of Theorem 1 in Section 1 of the Appendix by replacing the clustering output $\mathcal{C}(\cdot)$ with $\bigcap_{t=1}^T C^{(t)}(\cdot)$.
\end{proof}

It follows from Theorem~\ref{thm:truncnorm_kmeans} that computing the $p$-value $p_{j,\text{k-means}}$ reduces to characterizing the set $\mathcal{S}_j\Big(x; \bigcap_{t=1}^T C^{(t)}(x)\Big) $ in \eqref{eq:defS_kmeans}. It turns out we could directly leverage the results in \citet{chen2022selective} (reproduced below as Lemma~\ref{lemma:chen_kmeans}), which tests for a difference in the means for the entire vector, to arrive at a computationally-efficient recipe.

\begin{corollary}[Chen and Witten (2023)]
    \label{lemma:chen_kmeans}
For any $w \in \mathbb{R}^q$, define the set
{
\begin{equation} 
\mathcal{S}\Big(w; \bigcap_{t=1}^T C^{(t)}(x)\Big)  = 
\left \{ \phi \in \mathbb{R}: \bigcap_{t=1}^T \left \{ {C}^{(t)}\left(\left (I - \frac{\hat \nu \hat \nu^T}{\|\hat \nu\|_2^2} \right ) x  + \frac{ \phi \hat \nu w^T}{\|\hat \nu\|_2^2}\right ) =  {C}^{(t)}(x)\right \} \right\}.
\label{eq:kmeans_sw}
\end{equation}  
}
Suppose that we apply the $k$-means clustering algorithm (Algorithm~\ref{algo:k_means_alt_min}) to a matrix ${x}\in\mathbb{R}^{n\times q}$, to obtain $K$ clusters in at most $T$ steps. Then, the set $\mathcal{S}\Big(w; \bigcap_{t=1}^T C^{(t)}(x)\Big)$ defined in \eqref{eq:kmeans_sw} can be computed in $\mathcal{O}\big(nKT(n+q)+nKT\log(nKT)\big)$ operations.
\end{corollary}
\begin{proof}
   Corollary~\ref{lemma:chen_kmeans} follows from the observation Proposition 5 in \citet{chen2022selective} proves the special case of $\mathcal{S}\Big(w; \bigcap_{t=1}^T C^{(t)}(x)\Big)$ with $w = \nu^Tx/\Vert  \hat\nu^Tx \Vert_2$. However, no special properties of the vector $ \hat\nu^Tx/\Vert  \hat\nu^Tx \Vert_2$ was used in the proof, and replacing $ \hat\nu^Tx/\Vert  \hat\nu^Tx \Vert_2$ with $\Sigma_j/\Sigma_{jj}$ everywhere in their proof yields the desired result.
\end{proof}



\section{Additional simulation results}

In Section~\ref{subsec:power}, we compared the conditional power of the tests based on $p_{j, \text{k-means}}$, $p_{j, \text{average}}$, $p_{j, \text{centroid}}$, and $p_{j, \text{single}}$  under model \eqref{eq:mod}.

Here, we consider a different notion of power that does not condition on the estimated clusters $\hat{G}$ and $\hat{G}'$ being true clusters. In this case, comparing the power of the tests requires a bit of care, because the ``effect size'', $\Big|[\mu^T \hat\nu]_j \Big|$ may differ across simulated data sets from the same data-generating distribution, as the estimated clusters and the corresponding contrast vector $\hat\mu$ might vary across simulated data sets. As a result, we consider the power of the tests \emph{as a function of $\Big|[\mu^T \hat\nu]_j \Big|$}. We fit a regression spline using the \texttt{gam} function in the \texttt{R} package \texttt{mgcv}~\citep{wood2017} to obtain a smooth estimate of power on the same simulated data sets from Section~\ref{subsec:power}. 

The results are in Figure~\ref{fig:smooth_power}. The power of the tests that reject $\hat H_{0j}: \bar{\mu}_{\hat G j } =  \bar{\mu}_{\hat G' j}$ when $p_{j, \text{k-means}}$, $p_{j, \text{average}}$, $p_{j, \text{centroid}}$, or $p_{j, \text{single}}$ is less than $\alpha=0.05$ increases as $\Big|[\mu^T \hat\nu]_j \Big|$ increases. For a given value of $\Big|[\mu^T \hat\nu]_j \Big|$, the power of different tests depends on the correlation among features. When data have low to moderate correlations, the tests based on  $p_{j, \text{average}}$ and $p_{j, \text{centroid}}$ generally have higher power than the tests based on $p_{j, \text{k-means}}$ and $p_{j, \text{single}}$. When the correlation among features is high (e.g., $\rho=0.8$), the test based on $p_{j, \text{single}}$ has the highest power, followed by those based on $p_{j, \text{average}}$ and $p_{j, \text{centroid}}$; the test based on  $p_{j, \text{k-means}}$ has the lowest power. 
\begin{figure}[htbp!]
  \centering
\includegraphics[width=0.8\linewidth]{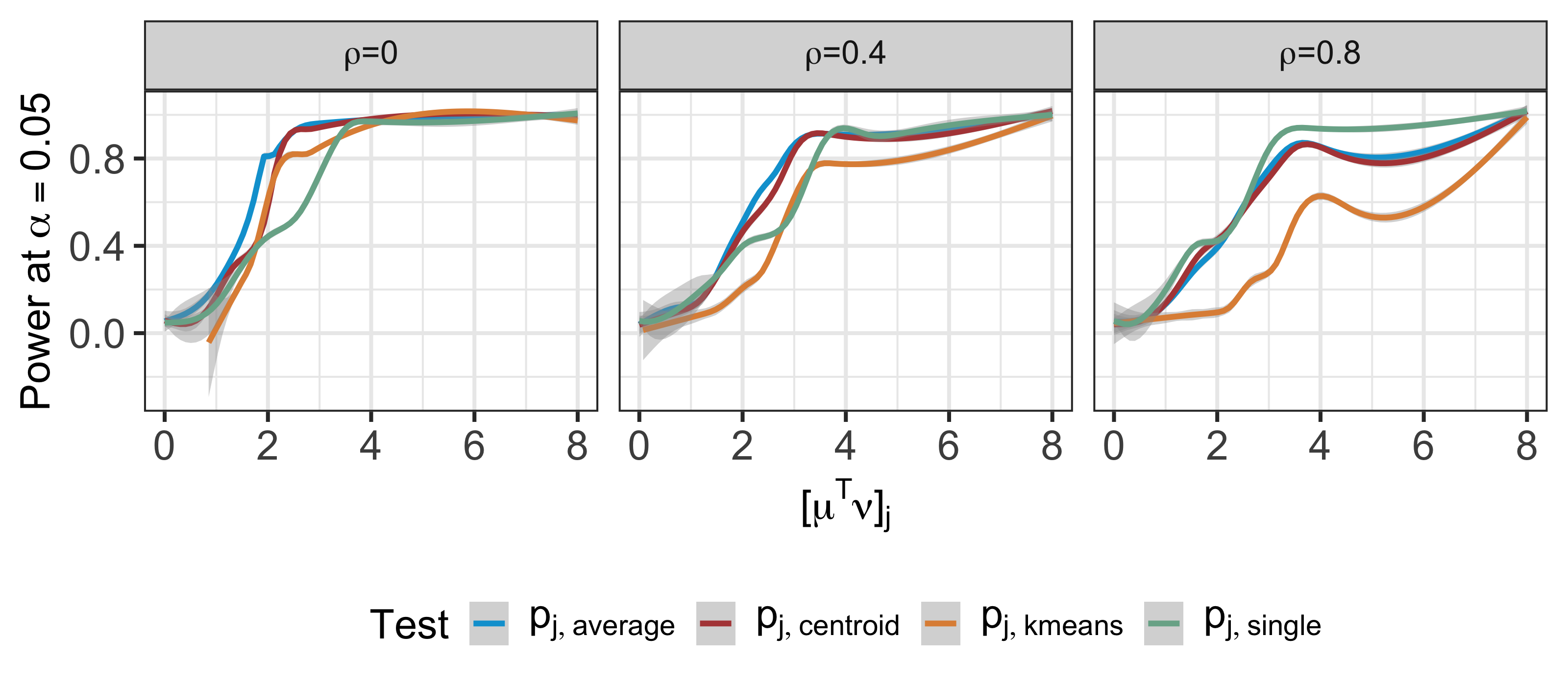}
\caption{We fit a regression spline on the simulated data sets from Section~\ref{subsec:power} to display the power of the tests based on $p_{j, \text{k-means}}$, $p_{j, \text{average}}$, $p_{j, \text{centroid}}$, and $p_{j, \text{single}}$.}
\label{fig:smooth_power}
\end{figure}


\end{document}